\documentclass[letterpaper,11pt]{article}

\newcommand{\qedllncs}{}

\newif\ifincludefigures
\includefigurestrue

\usepackage[letterpaper, left=2.5cm, right=2.5cm, top=2.5cm, bottom=2.5cm]{geometry}
\usepackage{amsmath}
\usepackage{amsthm}
\usepackage{mathtools} 
\usepackage{booktabs} % For formal tables
\usepackage{amsmath}
\usepackage{amssymb}
\usepackage{enumerate}
\usepackage{bm}
\usepackage{enumitem}
\newcommand{\ignore}[1]{{}}
\usepackage[labelfont=bf]{caption}

\usepackage{url}
\usepackage{hyperref}
\usepackage{stackrel}
\usepackage{mathtools}
\usepackage{algorithm}
\usepackage{algpseudocode}
\usepackage{tikz}
\usepackage{hyphenat}
\usepackage{longtable}
\usepackage{bm}

\renewcommand\vec{\bm} %vec -> bold

\algtext*{EndFor}% Remove "end while" text
\algtext*{EndIf}% Remove "end if" text
\newcommand\numberthis{\addtocounter{equation}{1}\tag{\theequation}} % numbering with \numberthis
\newcommand*\diff{\mathop{}\!\mathrm{d}}

\newcommand{\minuv}{\min}%{{\min(u,v)}}
\newcommand{\maxuv}{\max}%{{\max(u,v)}}

\newcounter{thm}
\newtheorem{theorem}[thm]{Theorem}
\newtheorem{lemma}[thm]{Lemma}
\newtheorem{corollary}[thm]{Corollary}
\begin{document}
	
	\title{Algorithms for Energy Conservation in Heterogeneous Data Centers\footnote{Work supported by the European Research Council, Grant Agreement No.\ 691672.}}
	
	\author{Susanne Albers \\
		Technical University of Munich \\
		albers@in.tum.de \\
		\and
		Jens Quedenfeld\footnote{Contact author} \\
		Technical University of Munich \\
		jens.quedenfeld@in.tum.de \\
	}
	
	\maketitle
	
	\begin{abstract}
		Power consumption is the major cost factor in data centers. It can be reduced by dynamically right-sizing the data center according to the currently arriving jobs. If there is a long period with low load, servers can be powered down to save energy. For identical machines, the problem has already been solved optimally by Lin et al. (2013) and Albers and Quedenfeld (2018).
		
		In this paper, we study how a data-center with heterogeneous servers can dynamically be right-sized to minimize the energy consumption. There are $d$ different server types with various operating and switching costs. 
		We present a deterministic online algorithm that achieves a competitive ratio of~$2d$ as well as a randomized version that is $1.58d$-competitive. Furthermore, we show that there is no deterministic online algorithm that attains a competitive ratio smaller than~$2d$. Hence our deterministic algorithm is optimal. In contrast to related problems like convex body chasing and convex function chasing%~\cite{FriedmanLinial1993,Sellke2020}
		, we investigate the discrete setting where the number of active servers must be integral, so we gain truly feasible solutions.
		
		%\keywords{Online algorithm \and Lower bound \and Approximation algorithm}
	\end{abstract}

\section{Introduction}
\label{sec:intro}
Energy management is an important issue in data centers. A huge amount of a data center's financial budget is spent on electricity that is needed to operate the servers as well as to cool them \cite{Brill2007,Hamilton2008}. 
However, server utilization is typically low. In fact there are data centers where the average server utilization is 
as low as 12\%~\cite{Delforge2014}; only for a few days a year is full processing power needed. Unfortunately, idle servers still consume about half of their peak power~\cite{Schmid2009power}. Therefore, right-sizing a data center by powering down idle servers can save a significant amount of energy. However, shutting down a server and powering it up immediately afterwards incurs much more cost than holding the server in the active state during this time period. The cost for powering up and down does not only contain the increased energy consumption but also, for example, wear-and-tear costs or the risk that the server does not work properly after restarting~\cite{LinWierman2013extended}. Consequently, algorithms are needed that manage the number of active servers to minimize the total cost, without knowing when new jobs will arrive in the future. 
%
%Since about 3\% of the global electricity production is consumed by data centers, right-sizing is not only important for economical but also for ecological reasons.
Since about 3\% of the global electricity production is consumed by data centers~\cite{Bawden2016}, a reduction of their energy consumption can also decrease greenhouse emissions. Thus, right-sizing data centers is not only important for economical but also for ecological reasons.

Modern data centers usually contain heterogeneous servers. If the capacity of a data center is no longer sufficient, it is extended by including new servers. The old servers are still used however. Hence, there are different server types with various operating and switching costs in a data center. Heterogeneous data centers may also include
different processing architectures. There can be servers that use GPUs to perform massive parallel calculations. However, GPUs are not suitable for all jobs. For example, tasks with many branches can be computed much faster on common CPUs than on GPUs~\cite{Shan2006}.

\textbf{Problem Formulation}
We consider a data center with $d$ different server types. There are $m_j$ servers of type $j$. Each server has an active state where it is able to process jobs, and an inactive state where no energy is consumed. Powering up a server of type $j$ (i.e., switching from the inactive into the active state) incurs a cost of~$\beta_j$ (called \emph{switching} cost); powering down does not cost anything. 
We consider a finite time horizon consisting of the time slots $\{1, \dots, T\}$.
For each time slot $t \in \{1, \dots, T\}$, jobs of total volume $\lambda_t \in \mathbb{N}_0$ arrive and have to be processed during the time slot. 
There must be at least $\lambda_t$ active servers to process the arriving jobs.
We consider a basic setting where the operating cost of a server of type $j$ is load- and time-independent and denoted by $l_{j} \in \mathbb{R}_{\geq 0}$. Hence, an active server incurs a constant but %individual,
type-dependent operating cost per time slot.

A schedule $X$ is a sequence $\vec{x}_1, \dots, \vec{x}_T$ with $\vec{x}_t = (x_{t,1}, \dots, x_{t,d})$ where each $x_{t,j}$ indicates the number of active servers of type $j$ during time slot $t$. At the beginning and the end of the considered time horizon all servers are shut down, i.e., $\vec{x}_0 = \vec{x}_{T+1} = (0, \dots, 0)$. A schedule is called \emph{feasible} if there are enough active servers to process the arriving jobs and if there are not more active servers than available, i.e., $\sum_{j=1}^{d} x_{t,j} \geq \lambda_t$ and $x_{t,j} \in \{0, 1, \dots, m_j\}$ for all $t \in \{1, \dots, T\}$ and $j \in \{1, \dots, d\}$. The cost of a feasible schedule is defined by
\begin{equation}\label{eqn:online:cost}
C(X) \coloneqq \sum_{t=1}^T \left( \sum_{j=1}^{d} l_{j} x_{t,j} + \sum_{j=1}^{d} \beta_j (x_{t,j} - x_{t-1, j})^+ \right) 
\end{equation}

where $(x)^+ \coloneqq \max (x, 0)$.
The switching cost is only paid for powering up. However, this is not a restriction, since all servers are inactive at the beginning and end of the workload. Thus the cost of powering down can be folded into the cost of powering up. 
A problem instance is specified by the tuple $\mathcal{I} = (T, d, \vec{m}, \vec{\beta}, \vec{l}, \Lambda)$ where $\vec{m} = (m_1, \dots, m_d)$, $\vec{\beta} = (\beta_1, \dots, \beta_d)$, $\vec{l} = (l_{1}, \dots, l_{d})$ and $\Lambda = (\lambda_1, \dots, \lambda_T)$.
The task is to find a schedule with minimum cost. 

We focus on the central case without \emph{inefficient} server types. A server type~$j$ is called \emph{inefficient} if there is another server type $j' \not= j$
with both smaller (or equal) operating and switching costs, i.e., $l_j \geq l_{j'}$ and $\beta_j \geq \beta_{j'}$. This assumption is natural because a better server type with a lower operating cost usually has a higher switching cost. An inefficient server of type $j$ is only powered up, if all servers of all types $j'$ with $\beta_{j'} \leq \beta_{j}$ and $l_{j'} \leq l_{j}$ are already running. Therefore, excluding inefficient servers is not a relevant restriction in practice. In related work, Augustine et al. \cite{AugustineIrani2008} exclude inefficient states when operating a single server.
%A schedule $X$ is called \emph{optimal}, if there is no other schedule $Y$ with $C(Y) < C(X)$. 

\textbf{Our contribution}
We analyze the online setting of this problem where the job volumes $\lambda_t$ arrive one-by-one. The vector of the active servers $\vec{x}_t$ has to be determined without knowledge of future jobs~$\lambda_{t'}$ with $t' > t$. A main contribution of our work, compared to previous results, is that we investigate heterogeneous
data centers and examine the online setting when truly feasible (integral) solutions are sought.

In Section~\ref{sec:online:const}, we present a $2d$-competitive deterministic online algorithm, i.e., the total cost of the schedule calculated by our algorithm is at most $2d$ times larger than the cost of an optimal offline solution. 
%Our algorithm assumes a data center without \emph{inefficient} server types. A server type $j$ is called \emph{inefficient} if there is another server type $j' \not= j$
%with both smaller (or equal) operating and switching costs, i.e., $l_j \geq l_{j'}$ and $\beta_j \geq \beta_{j'}$. This assumption is natural because a better server type with a lower operating cost usually has a higher switching cost. An inefficient server of type $j$ is only powered up, if all servers of all types $j'$ with $\beta_{j'} \leq \beta_{j}$ and $l_{j'} \leq l_{j}$ are already running. Therefore, excluding inefficient servers is not a relevant restriction in practice. In related work, Augustine et al. \cite{AugustineIrani2008} exclude inefficient states when operating a single server.
%
Roughly, our algorithm works as follows. It calculates an optimal schedule for the jobs received so far and ensures that the operating cost of the active servers is at most as large as the operating cost of the active servers in the optimal schedule. If this is not the case, servers with high operating cost are replaced by servers with low operating cost. If a server is not used for a specific duration depending on its switching and operating costs, it is shut down.

In Section~\ref{sec:online:rand}, we devise a randomized version of our algorithm achieving a competitive ratio of $\frac{e}{e-1}d \approx 1.582d$ against an oblivious adversary.

%In Section~\ref{sec:online:rand}, we devise a randomized version of our algorithm achieving a competitive ratio of $e / (e-1) \cdot d \approx 1.582d$ against an oblivious adversary.

In Section~\ref{sec:online:lower}, we show that there is no deterministic online algorithm that achieves a competitive ratio smaller than $2d$. Therefore, our algorithm is optimal. 
Additionally, for a data center that contains $m$ unique servers (that is $m_j = 1$ for all $j \in \{1, \dots, d\}$), we show that the best achievable competitive ratio is $2m$. 

\textbf{Related work} The design of energy-effcient algorithms has received quite some research interest over the last years, see e.g. \cite{BansalKimbrelPruhs2007,IraniPruhs2005,Antoniadis2020} and references therein. Specifically, data center right-sizing has attracted considerable attention lately.
%Data-center right-sizing has received considerable research attention in recent years.
Lin and Wierman \cite{LinWierman2013,LinWierman2013extended} analyzed the data-center right-sizing problem for data centers with identical servers ($d=1$). The operating cost is load-dependent and modeled by a convex function. In contrast to our setting, continuous solutions are allowed, i.e., the number of active server $x_{t}$ can be fractional. This allows for other techniques in the design and analysis of an algorithm, but the created schedules cannot be used directly in practice. They gave a 3-competitive deterministic online algorithm for this problem. Bansal et al.~\cite{Bansal2015} improved this result by randomization and developed a 2-competitive online algorithm. In our previous paper~\cite{AlbersQuedenfeld2018}, we showed that~2 is a lower bound for randomized algorithms in the continuous setting; this result was independently shown by~\cite{Antoniadis2017}. Furthermore, we analyzed the discrete setting of the problem where the number of active servers is integral ($x_t \in \mathbb{N}_0$). We presented a 3-competitive deterministic and a 2-competitive randomized online algorithm. Moreover, we proved that these competitive ratios are optimal. %--%Chen et al.~\cite{Chen2015} studied how predictions for the future can be used to create better schedules. 

Data-center right-sizing of heterogeneous data centers is related to convex function chasing, which is also known as smoothed online convex optimization~\cite{ChenGoelWierman2018}. At each time slot $t$, a convex function $f_t$ arrives. The algorithm then has to choose a point $\vec{x}_t$ and pay the cost $f_t(\vec{x}_t)$ as well as the movement cost $\|\vec{x}_t - \vec{x}_{t-1}\|$ where $\|\cdot \|$ is any metric.
The problem described by equation~\eqref{eqn:online:cost} is a special case of convex function chasing if fractional schedules are allowed, i.e., $x_{t,j} \in [0, m_j]$ instead of $x_{t,j} \in \{0, \dots, m_j\}$. The operating cost $\sum_{j=1}^{d} l_{j} x_{t,j}$ in equation~\eqref{eqn:online:cost} together with the feasibility requirements can be modeled as  a convex function that is infinite for $\sum_{j=1}^{d} x_{t,j} < \lambda_t$ and $x_{t,j} \notin [0,m_j]$. The switching cost equals the Manhattan metric if the number of servers is scaled appropriately. Sellke \cite{Sellke2020} gave a $(d+1)$-competitive algorithm for convex function chasing. 
%Note that the best achieveable competitive ratio in the discrete setting is $2d$.
%Note that the discrete setting is much more difficult due to the lower bound of $2d$. % is a lower bound in the discrete setting. 
%
%Note that we showed that $2d$ is a lower bound in the discrete setting. 
A similar result was found by Argue et al.~\cite{Argue2020}.

In the discrete setting, convex function chasing has at least an exponential competitive ratio, as the following setting shows. Let $m_j = 1$ and $\beta_j = 1$ for all $j \in \{1, \dots, d\}$, so the possible server configurations are $\{0,1\}^d$. The arriving convex functions $f_t$ are infinite for the current position $\vec{x}_{t-1}$ of the online algorithm and $0$ for all other positions $\{0,1\}^d \setminus \{\vec{x}_{t-1}\}$. After $T \coloneqq 2^d -1$ functions arrived, the switching cost paid by the algorithm is at least $2^d -1$ (otherwise it has to pay infinite operating costs), whereas the offline schedule can go directly to a position without any operating cost and only pays a switching cost of at most~$d$. 

Already for the 1-dimensional case (i.e. identical machines), it is not trivial to round a fractional schedule without increasing the competitive ratio (see~\cite{LinWierman2013extended} and \cite{AlbersQuedenfeld2018extended}). In $d$-dimensional space, it is completely unclear, if continuous solutions can be rounded without arbitrarily increasing the total cost.
%copied from itcs2021
Simply rounding up can lead to arbitrarily large switching costs, for example if the fractional solution rapidly switches between 1 and $1 + \epsilon$. Using a randomized rounding scheme like in \cite{AlbersQuedenfeld2018extended} (that was used for homogeneous data centers) independently for each dimension can result in an infeasible schedule (for example, if $\lambda_t = 1$ and $\vec{x}_t = (1/d, \dots, 1/d)$ is rounded down to $(0, \dots, 0)$).
Therefore, Sellke's result does not help us for analyzing the discrete setting. Other publications handling convex function chasing or convex body chasing are %\cite{BubeckSellke2020nested,Argue2020}.
\cite{FriedmanLinial1993,BansalBoehm2018,BubeckSellke2020nested}.

Goel and Wierman \cite{GoelWierman2018} developed a $(3+\mathcal{O}(1/\mu))$-competitive algorithm called Online Balanced Descent (OBD) for convex function chasing, where the arriving functions were required to be $\mu$-strongly convex. We remark that the operating cost defined by equation~\eqref{eqn:online:cost} is not strongly convex, i.e., $\mu = 0$. Hence their result cannot be used for our problem.	
A similar result is given by Chen et al. \cite{ChenGoelWierman2018} who showed that OBD is $(3+\mathcal{O}(1/\alpha))$-competitive if the arriving functions are locally $\alpha$-polyhedral. In our case, $\alpha = \min_{j \in \{1, \dots, d\}} l_j / \beta_j$, so $\alpha$ can be arbitrarily small depending on the problem instance.

%The well-known ski rental problem is a special case of data-center right-sizing with only one server (i.e., $d=1$ and $m_1 = 1$). The start and end state are $\vec{x}_0 = \vec{x}_{T+1} = (1)$. The switching cost is the cost of buying skis and the operating cost is the cost of renting them. For $t < T$, the job volume is zero, only $\lambda_T $

Another similar problem is the Parking Permit Problem by Meyerson~\cite{Meyerson2005}. There are $d$ different permits which can be purchased for $\beta_j$ dollars and have a duration of $D_j$ days. Certain days are driving days where at least one parking permit is needed ($\lambda_t \in \{0,1\}$). The permit cost corresponds to our switching cost. However, the duration of the permit is fixed to $D_j$, whereas in our problem the online algorithm can choose for each time slot if it wants to power down a server. Furthermore, there is no operating cost. Even if each server type is replaced by an infinite number of permits with the duration $t$ and the cost $\beta_j +  l_j \cdot t$, it is still a different problem, because the algorithm has to choose the time slot for powering down in advance (when the server is powered up). 

Data-center right-sizing of heterogeneous data centers is related to geographical load balancing analyzed in \cite{LinLiuWierman2012} and \cite{LiuLinWierman2011}. 
Other applications are shown in \cite{Wang2014,Kim2014,Chen2015,Kim2015,BadieiLiWierman2015,GoelChenWierman2017,Zhang2018}.

\section*{Notation}
\label{sec:notation}
Let $[k] \coloneqq \{1, 2, \dots k\}$, $[k]_0 \coloneqq \{0, 1, \dots k\}$ and $[k:l] \coloneqq \{k, k+1, \dots, l\}$ where $k,l \in \mathbb{N}_0$.

\section{Deterministic Online Algorithm}
\label{sec:online:const}

In this section we present a deterministic $2d$-competitive online algorithm for the problem described in the preceding section. % assuming that there are no inefficient servers.
The basic idea of our algorithm is to calculate an optimal schedule for the problem instance that ends at the current time slot. Based on this schedule, we decide when a server is powered up. If a server is idle for a specific time, it is powered down.

Formally, given the original problem instance $\mathcal{I} = (T, d, \vec{m}, \vec{\beta}, \vec{l}, \Lambda)$, the shortened problem instance~$\mathcal{I}^t$ is defined by $\mathcal{I}^t \coloneqq (t, d, \vec{m}, \vec{\beta}, \vec{l}, \Lambda^t)$ with $\Lambda^t = (\lambda_1, \dots, \lambda_t)$. Let $\hat{X}^t$ denote an optimal schedule for $\mathcal{I}^t$ and let $X^\mathcal{A}$ be the schedule calculated by our algorithm~$\mathcal{A}$. 

W.l.o.g. there are no server types with the same operating and switching costs, i.e., $\beta_j = \beta_{j'}$ and $l_j = l_{j'}$ implies $j = j'$. Furthermore, let $l_1 > \dots > l_d$, i.e., the server types are sorted by their operating costs. Since inefficient server types are excluded, this implies that $\beta_1 < \dots < \beta_d$. 

%Let $[n] \coloneqq \{1, \dots, n\}$ where $n \in \mathbb{N}$. 
We separate a problem instance into $m \coloneqq \sum_{j=1}^{d} m_j$ lanes. At time slot $t$, there is a single job in lane $k \in [m]$, if and only if $k \leq \lambda_t$. We can assume that $\lambda_t \leq m$ holds for all $t \in [T]$, because otherwise there is no feasible schedule for the problem instance. 
Let $X$ be an arbitrary feasible schedule with $\vec{x}_t = (x_{t,1}, \dots, x_{t,d})$. We define
\begin{equation} \label{eqn:online:const:ytkdef}
y_{t,k} \coloneqq  \begin{cases}
\max \{j \in [d] \mid \sum_{j' = j}^{d} x_{t,j'} \geq k\}  &\text{if $k \in \left[\sum_{j=1}^{d} x_{t,j}\right]$} \\
0 &\text{else}
\end{cases}
\end{equation}
to be the server type that handles the $k$-th lane during time slot $t$ (see Figure~\ref{fig:online:const:xtjytk}). If $y_{t,k} = 0$, then there is no active server in lane $k$ at time slot $t$. By definition, the values $y_{t,1}, \dots, y_{t,m}$ are sorted in descending order, i.e., $y_{t,k} \geq y_{t,k'}$ for $k < k'$. 
Note that $y_{t,k} = 0$ implies $\lambda_t < k$, because otherwise there are not enough active servers to handle the jobs at time $t$. For the schedule~$\hat{X}^t$, the server type used in lane $k$ at time slot $t'$ is denoted by $\hat{y}^t_{t',k}$. Our algorithm calculates $y^{\mathcal{A}}_{t,k}$ directly, the corresponding variables $x^\mathcal{A}_{t,j}$ can be determined by $x^\mathcal{A}_{t,j} = |\{k \in [m] \mid y^{\mathcal{A}}_{t,k} = j\}|$. %In contrast to the definition above, it is possible that the sequence $y^{\mathcal{A}}_{t,1}, \dots, y^{\mathcal{A}}_{t,m}$ is not sorted. 
A tabular overview of the notation is shown in~\ref{sec:appendix:variables}.

\ifincludefigures
\begin{figure}[t]
	\setlength{\abovecaptionskip}{6pt plus 0pt minus 0pt}
	\setlength{\belowcaptionskip}{0pt plus 0pt minus 0pt}
	
	\centering
	\begin{tikzpicture}

	\def\lambdaarray{{0,1,3,0,0,2,1,4,0,3,1,1,0}}
	\def\xtja{{0,0,2,0,0,0,0,1,0,0,0,0,0}}
	\def\xtjb{{0,0,0,0,0,1,1,2,2,2,0,0,0}}
	\def\xtjc{{0,1,1,1,1,1,1,1,1,1,1,1,0}}
	
	\pgfmathsetmacro{\xBegin}{0}
	\pgfmathsetmacro{\xStep}{0.65}
	\pgfmathsetmacro{\xLength}{12}
	\pgfmathsetmacro{\tMax}{\xLength - 1}
	
	\pgfmathsetmacro{\yBeginLambda}{0}
	\pgfmathsetmacro{\yBeginXtj}{-1.2}
	\pgfmathsetmacro{\yStep}{0.55}
	\pgfmathsetmacro{\yLength}{4}
	\pgfmathsetmacro{\yBeginYtk}{\yBeginXtj - \yStep * \yLength -1.1}

	\pgfmathsetmacro{\axisExtend}{0.6}
	
	\pgfmathsetmacro{\tickHalfLength}{0.07}
	
	\tikzstyle{arrow}=[-stealth]
	\tikzstyle{job}=[very thick]
	\tikzstyle{jobfill}=[fill=white]%[fill=red!20!white]
	\tikzstyle{schedule}=[]%[fill=green!20!white]
	
	\foreach \x in {0,...,\tMax} {
		\fill[jobfill] (\xBegin + \x * \xStep , \yBeginLambda) rectangle
			(\xBegin + \x * \xStep + \xStep, \yBeginLambda + \lambdaarray[\x] * \yStep);
	}
	
	\draw[arrow] (\xBegin, \yBeginLambda) to (\xBegin + \xStep * \xLength +  \axisExtend * \xStep, \yBeginLambda)
		node[below] {$t$};
	\foreach \x in {1,...,\xLength } {
		\draw (\xBegin + \x * \xStep, \yBeginLambda - \tickHalfLength) to (\xBegin + \x * \xStep, \yBeginLambda + \tickHalfLength);
		\pgfmathsetmacro{\t}{int(\x - 1)}
		\node[below] at (\xBegin + \x * \xStep - 0.5 * \xStep, \yBeginLambda) {\t};
	}
	
	\draw[arrow] (\xBegin, \yBeginLambda) to (\xBegin, \yBeginLambda  + \yStep * \yLength + \axisExtend * \yStep)
		node[left] {$\lambda_t$};
	\foreach \y in {1,...,\yLength } {
		\draw (\xBegin - \tickHalfLength, \yBeginLambda + \y * \yStep) to (\xBegin + \tickHalfLength, \yBeginLambda + \y * \yStep);
	}
	
	\foreach \y in {1,...,\yLength } {
		\node[left] at (\xBegin - \tickHalfLength, \yBeginLambda + \y * \yStep) {\y};
	}

	\foreach \x in {0,...,\tMax} {
		\pgfmathsetmacro{\m}{\lambdaarray[\x] - 1}
		\pgfmathtruncatemacro\myresult{\m>=1?1:0}
		\ifnum\myresult>0\relax
			\foreach \y in {1,...,\m} {
				\draw[job, thin] (\xBegin + \x * \xStep, \yBeginLambda + \y * \yStep) 
				to (\xBegin + \x * \xStep + \xStep, \yBeginLambda + \y * \yStep);
			}
		\fi
		
		\draw[job] (\xBegin + \x * \xStep , \yBeginLambda + \lambdaarray[\x] * \yStep) 
				to (\xBegin + \x * \xStep + \xStep, \yBeginLambda + \lambdaarray[\x] * \yStep)
				to (\xBegin + \x * \xStep + \xStep, \yBeginLambda + \lambdaarray[\x + 1] * \yStep)
				to (\xBegin + \x * \xStep + \xStep + 0.01 * \xStep, \yBeginLambda + \lambdaarray[\x + 1] * \yStep);
	}

	\node[left] at (\xBegin - 0 * \xStep, \yBeginXtj) {$\vec{x}_t = $};
	\foreach \x in {0,...,\tMax} {
		\pgfmathsetmacro{\xtjax}{\xtja[\x] }
		\pgfmathsetmacro{\xtjbx}{\xtjb[\x] }
		\pgfmathsetmacro{\xtjcx}{\xtjc[\x] }
		\node at (\xBegin + \x * \xStep + 0.5 * \xStep, \yBeginXtj) 
			{$\begin{pmatrix} \xtjax \\ \xtjbx \\ \xtjcx \end{pmatrix}$};
	}

	\draw[arrow] (\xBegin, \yBeginYtk) to (\xBegin + \xStep * \xLength +  \axisExtend * \xStep, \yBeginYtk)
		node[below] {$t$};
	\foreach \x in {1,...,\xLength } {
		\draw (\xBegin + \x * \xStep, \yBeginYtk - \tickHalfLength) to (\xBegin + \x * \xStep, \yBeginYtk + \tickHalfLength);
		\pgfmathsetmacro{\t}{int(\x - 1)}
		\node[below] at (\xBegin + \x * \xStep - 0.5 * \xStep, \yBeginYtk) {\t};
	}
	
	\draw[arrow] (\xBegin, \yBeginYtk) to (\xBegin, \yBeginYtk  + \yStep * \yLength +  \axisExtend * \yStep)
		node[left] {$k$};
	\foreach \y in {1,...,\yLength } {
		\draw (\xBegin - \tickHalfLength, \yBeginYtk + \y * \yStep) to (\xBegin + \tickHalfLength, \yBeginYtk + \y * \yStep);
	}
	\foreach \y in {1,...,\yLength } {
		\node[left] at (\xBegin - \tickHalfLength, \yBeginYtk + \y * \yStep) {\y};
	}
	
%	\foreach \x in {0,...,\tMax} {
%		\draw[line width=2pt, color=red] (\xBegin + \x * \xStep , \yBeginYtk + \lambdaarray[\x] * \yStep) 
%		to (\xBegin + \x * \xStep + \xStep, \yBeginYtk + \lambdaarray[\x] * \yStep)
%		to (\xBegin + \x * \xStep + \xStep, \yBeginYtk + \lambdaarray[\x + 1] * \yStep)
%		to (\xBegin + \x * \xStep + \xStep + 0.5 * \xStep, \yBeginYtk + \lambdaarray[\x + 1] * \yStep);
%	}
	
	\draw[schedule] (\xBegin + \xStep, \yBeginYtk) rectangle (\xBegin + 12 * \xStep, \yBeginYtk + \yStep)  
		node[pos=.5] {$3$};
	\draw[schedule] (\xBegin + 2 * \xStep, \yBeginYtk  + \yStep) rectangle (\xBegin + 3 * \xStep, \yBeginYtk + 2 * \yStep)  
		node[pos=.5] {$1$};
	\draw[schedule] (\xBegin + 2 * \xStep, \yBeginYtk  + 2* \yStep) rectangle (\xBegin + 3 * \xStep, \yBeginYtk + 3* \yStep)  
		node[pos=.5] {$1$};
	\draw[schedule] (\xBegin + 5 * \xStep, \yBeginYtk  + \yStep) rectangle (\xBegin + 10 * \xStep, \yBeginYtk + 2 * \yStep)  
		node[pos=.5] {$2$};
	\draw[schedule] (\xBegin + 7 * \xStep, \yBeginYtk  + 2 * \yStep) rectangle (\xBegin + 10 * \xStep, \yBeginYtk + 3* \yStep)  
		node[pos=.5] {$2$};
	\draw[schedule] (\xBegin + 7 * \xStep, \yBeginYtk  + 3 * \yStep) rectangle (\xBegin + 8 * \xStep, \yBeginYtk + 4* \yStep)  
		node[pos=.5] {$1$};
		
	\end{tikzpicture}
	\caption{Example of a job sequence (upper plot) and a feasible schedule $X$ written in both notations $x_{t,j}$ (middle) and $y_{t,k}$ (lower plot). \normalfont Outside of the rectangles in the lower plot, the value of $y_{t,k}$ is 0.}
	\label{fig:online:const:xtjytk}
\end{figure}
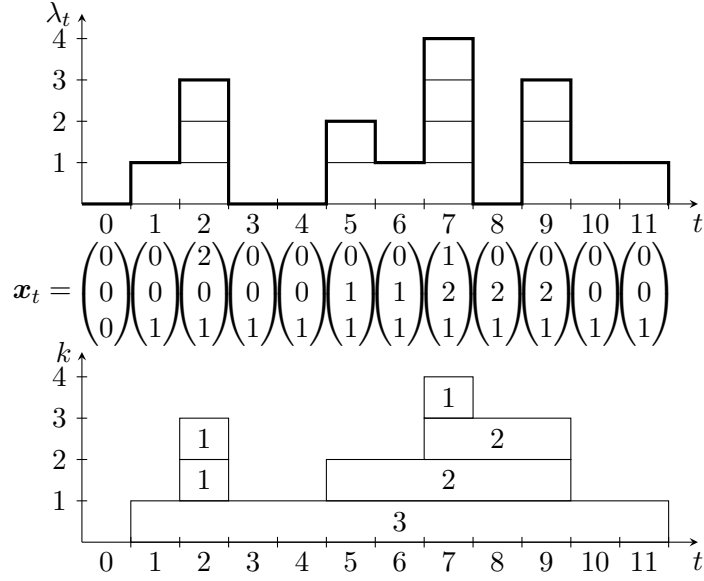
\fi

Our algorithm works as follows:
First, an optimal solution $\hat{X}^t$ is calculated. If there are several optimal schedules, we choose a schedule that fulfills the inequality $\hat{y}^t_{t',k} \geq \hat{y}^{t-1}_{t',k}$ for all time slots $t' \in [t]$ and lanes $k \in [m]$, so $\hat{X}^t$ never uses smaller server types than the previous schedule $\hat{X}^{t-1}$. We will see in Lemma~\ref{lemma:online:const:opt:stypes} that such a schedule exists and how to construct it.

If there is a server type $j$ with $l_j = 0$, then in an optimal schedule such a server can be powered up before it is needed, although $\lambda_t = 0$ holds for this time slot. Similarly, such a server can run for more time slots than necessary. W.l.o.g. let $\hat{X}^t$ be a schedule where servers are powered up as late as possible and powered down as early as possible. 

Beginning from the lowest lane ($k = 1$), it is ensured that $\mathcal{A}$ uses a server type that is not smaller than the server type used by $\hat{X}^t$, i.e., $y^{\mathcal{A}}_{t,k} \geq \hat{y}^t_{t,k}$ must be fulfilled. If the server type $y^{\mathcal{A}}_{t-1,k}$ used in the previous time slot is smaller than $\hat{y}^t_{t,k}$, it is powered down and server type $\hat{y}^t_{t,k}$ is powered up. 
A server of type $j$ that is not replaced by a greater server type stays active for $\bar{t}_j \coloneqq \left\lfloor {\beta_j}/{l_j} \right\rfloor$ time slots.
%A server of type $j$ is powered down  after $\bar{t}_j \coloneqq \left\lfloor {\beta_j}/{l_j} \right\rfloor$ time slots (i.e., it stays active for $\bar{t}_j$ time slots). 
If $\hat{X}^t$ uses a smaller server type $j' \leq j$ in the meantime, then server type~$j$ will run for at least $\bar{t}_{j'}$ further time slots (including time slot~$t$). Formally, a server of type $j$ in lane $k$ is powered down at time slot $t$, if $\hat{y}^{t'}_{t',k} \not = j'$ holds for all server types $j' \leq j$ and time slots $t' \in [t - \bar{t}_{j'} + 1 : t]$.% with $[a:b] \coloneqq \{a, a+1, \dots, b\}$.  

The pseudocode below clarifies how algorithm~$\mathcal{A}$ works. The variables $e_k$ for $k \in [m]$ store the time slot when the server in the corresponding lane will be powered down. Figure~\ref{fig:online:const:algo} visualizes how the schedule $X^{\mathcal{A}}$ changes from time slot $t-1$ to $t$.

%\vspace{-17pt}

\begin{algorithm}\label{alg:online:const}
	\caption{Algorithm $\mathcal{A}$}
	%\footnotesize
	\begin{algorithmic}[1]
		\For{$t \coloneqq 1$ \textbf{to} $T$}
			\State Calculate $\hat{X}^t$ such that $\hat{y}^t_{t',k} \geq \hat{y}^{t-1}_{t',k}$ for all $t' \in [t]$ and $k \in [m]$
			\For{$k \coloneqq 1$ \textbf{to} $m$}
				\If {$y^{\mathcal{A}}_{t-1,k} < \hat{y}^t_{t,k}$ \textbf{or} $t \geq e_k$}
					\State $y^{\mathcal{A}}_{t,k} \coloneqq \hat{y}^t_{t,k}$
					\State $e_k \coloneqq t + \bar{t}_{y^{\mathcal{A}}_{t,k}}$
				 \Else 
				 	\State $y^{\mathcal{A}}_{t,k} \coloneqq y^{\mathcal{A}}_{t-1,k}$
				 	\State $e_k \coloneqq \max \{e_k, t + \bar{t}_{\hat{y}^t_{t,k}}\}$ where $\bar{t}_0 \coloneqq 0$
				 \EndIf
			\EndFor
		\EndFor
	\end{algorithmic}
	\vspace{-3pt}
\end{algorithm}

%\vspace{-17pt}

\ifincludefigures
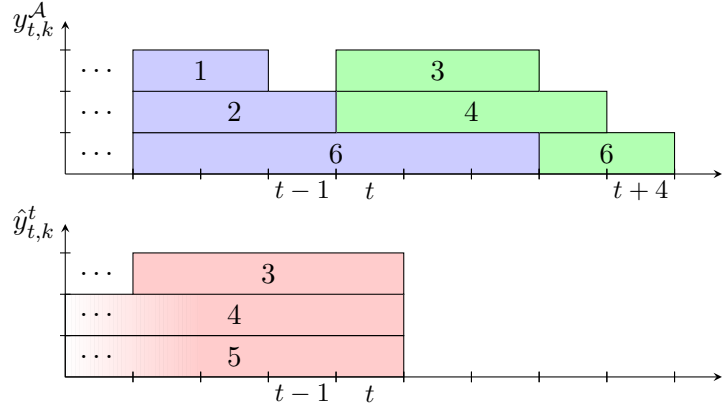
\begin{figure}[t]
	\setlength{\abovecaptionskip}{6pt plus 0pt minus 0pt}
	\setlength{\belowcaptionskip}{0pt plus 0pt minus 0pt}
	
	\centering
	\begin{tikzpicture}
		\pgfmathsetmacro{\xBegin}{0}
		\pgfmathsetmacro{\xStep}{0.9}
		\pgfmathsetmacro{\xLength}{9}
		\pgfmathsetmacro{\tMax}{\xLength - 1}
		
		\pgfmathsetmacro{\yBeginAlgo}{0}
		\pgfmathsetmacro{\yBeginOpt}{-2.7}
		\pgfmathsetmacro{\yStep}{0.55}
		\pgfmathsetmacro{\yLength}{3}
		
		\pgfmathsetmacro{\axisExtend}{0.7}
		
		\pgfmathsetmacro{\tickHalfLength}{0.07}
		
		\tikzstyle{arrow}=[-stealth]
		\tikzstyle{algbefore}=[fill=blue!20]
		\tikzstyle{algafter}=[fill=green!30]
		\tikzstyle{opt}=[fill=red!20]

		\draw[algafter] (\xBegin +7 * \xStep, \yBeginAlgo + 0 * \yStep) 
			rectangle (\xBegin + 9 * \xStep, \yBeginAlgo + 1 * \yStep) 
			node[pos=0.5] {$6$};
			
		\draw[algafter] (\xBegin + 4 * \xStep, \yBeginAlgo + 1 * \yStep) 
			rectangle (\xBegin + 8 * \xStep, \yBeginAlgo + 2* \yStep) 
			node[pos=0.5] {$4$};
			
		\draw[algafter] (\xBegin + 4 * \xStep, \yBeginAlgo + 2* \yStep) 
			rectangle (\xBegin + 7 * \xStep, \yBeginAlgo + 3 * \yStep) 
			node[pos=0.5] {$3$};

		\draw[algbefore] (\xBegin + 1 * \xStep, \yBeginAlgo + 0 * \yStep) 
			rectangle (\xBegin + 7 * \xStep, \yBeginAlgo + 1 * \yStep) 
			node[pos=0.5] {$6$};
			
		\draw[algbefore] (\xBegin + 1 * \xStep, \yBeginAlgo + 1 * \yStep) 
			rectangle (\xBegin + 4 * \xStep, \yBeginAlgo + 2 * \yStep) 
			node[pos=0.5] {$2$};
		
		\draw[algbefore] (\xBegin + 1 * \xStep, \yBeginAlgo + 2 * \yStep) 
			rectangle (\xBegin + 3 * \xStep, \yBeginAlgo + 3 * \yStep) 
			node[pos=0.5] {$1$};
			
		\node at (\xBegin + 0.5 * \xStep, \yBeginAlgo + 0.5 * \yStep) {$\dots$};
		\node at (\xBegin + 0.5 * \xStep, \yBeginAlgo + 1.5 * \yStep) {$\dots$};
		\node at (\xBegin + 0.5 * \xStep, \yBeginAlgo + 2.5 * \yStep) {$\dots$};
		
		\draw[arrow] (\xBegin, \yBeginAlgo) to (\xBegin + \xStep * \xLength +  \axisExtend * \xStep, \yBeginAlgo);
		\foreach \x in {1,...,\xLength } {
			\draw (\xBegin + \x * \xStep, \yBeginAlgo - \tickHalfLength) to (\xBegin + \x * \xStep, \yBeginAlgo + \tickHalfLength);
		}
		\foreach \x in {1,...,\xLength } {
			\draw (\xBegin + \x * \xStep, \yBeginAlgo - \tickHalfLength) to (\xBegin + \x * \xStep, \yBeginAlgo + \tickHalfLength);
		}
		\node[] at (\xBegin + 3.5 * \xStep, \yBeginAlgo - 0.4 * \yStep) {\small $t-1$};
		\node[] at (\xBegin + 4.5 * \xStep, \yBeginAlgo - 0.4 * \yStep) {\small $t$};
		\node[] at (\xBegin + 8.5 * \xStep, \yBeginAlgo - 0.4 * \yStep) {\small $t+4$};
		
		\draw[arrow] (\xBegin, \yBeginAlgo) to (\xBegin, \yBeginAlgo  + \yStep * \yLength + \axisExtend * \yStep)
		node[left] {$y^{\mathcal{A}}_{t,k}$};
		\foreach \y in {1,...,\yLength } {
			\draw (\xBegin - \tickHalfLength, \yBeginAlgo + \y * \yStep) to (\xBegin + \tickHalfLength, \yBeginAlgo + \y * \yStep);
		}
		
		\fill[left color=white, right color=red!20!white] 
			(\xBegin + 0 * \xStep, \yBeginOpt + 0 * \yStep) 
			rectangle (\xBegin + 2* \xStep, \yBeginOpt + 1 * \yStep);
		\fill[left color=white, right color=red!20!white] 
			(\xBegin + 0 * \xStep, \yBeginOpt + 1 * \yStep) 
			rectangle (\xBegin + 2* \xStep, \yBeginOpt + 2 * \yStep);
		\fill[color=red!20!white] 
			(\xBegin + 2 * \xStep, \yBeginOpt + 0 * \yStep) 
			rectangle (\xBegin + 5* \xStep, \yBeginOpt + 1 * \yStep);
		\fill[color=red!20!white] 
			(\xBegin + 2 * \xStep, \yBeginOpt + 1 * \yStep) 
			rectangle (\xBegin + 5* \xStep, \yBeginOpt + 2 * \yStep);
		
		\draw[] (\xBegin + 0 * \xStep, \yBeginOpt + 0 * \yStep) 
			rectangle (\xBegin + 5 * \xStep, \yBeginOpt + 1 * \yStep) 
			node[pos=0.5] {$5$};
		\draw[] (\xBegin + 0 * \xStep, \yBeginOpt + 1 * \yStep) 
			rectangle (\xBegin + 5 * \xStep, \yBeginOpt +2 * \yStep) 
			node[pos=0.5] {$4$};
		\draw[opt] (\xBegin + 1 * \xStep, \yBeginOpt + 2 * \yStep) 
			rectangle (\xBegin + 5 * \xStep, \yBeginOpt +3 * \yStep) 
			node[pos=0.5] {$3$};
		
		\node at (\xBegin + 0.5 * \xStep, \yBeginOpt + 0.5 * \yStep) {$\dots$};
		\node at (\xBegin + 0.5 * \xStep, \yBeginOpt + 1.5 * \yStep) {$\dots$};
		\node at (\xBegin + 0.5 * \xStep, \yBeginOpt + 2.5 * \yStep) {$\dots$};
		
		\draw[arrow] (\xBegin, \yBeginOpt) to (\xBegin + \xStep * \xLength +  \axisExtend * \xStep, \yBeginOpt);
		\foreach \x in {1,...,\xLength } {
			\draw (\xBegin + \x * \xStep, \yBeginOpt - \tickHalfLength) to (\xBegin + \x * \xStep, \yBeginOpt + \tickHalfLength);
		}
		\foreach \x in {1,...,\xLength } {
			\draw (\xBegin + \x * \xStep, \yBeginOpt - \tickHalfLength) to (\xBegin + \x * \xStep, \yBeginOpt + \tickHalfLength);
		}
		\node[] at (\xBegin + 3.5 * \xStep, \yBeginOpt - 0.4 * \yStep) {\small $t-1$};
		\node[] at (\xBegin + 4.5 * \xStep, \yBeginOpt - 0.4 * \yStep) {\small $t$};
		
		\draw[arrow] (\xBegin, \yBeginOpt) to (\xBegin, \yBeginOpt  + \yStep * \yLength + \axisExtend * \yStep)
		node[left] {$\hat{y}^{t}_{t,k}$};
		\foreach \y in {1,...,\yLength } {
			\draw (\xBegin - \tickHalfLength, \yBeginOpt + \y * \yStep) to (\xBegin + \tickHalfLength, \yBeginOpt + \y * \yStep);
		}

	\end{tikzpicture}
	\caption{{\normalfont (figure is colored)} Example of an update in algorithm~$\mathcal{A}$ from time slot~$t-1$ to~$t$. \normalfont The schedule of $\mathcal{A}$ (upper plot) at $t-1$ is shown in blue, the changes after reacting to $\lambda_t$ are printed in green. The optimal schedule $\hat{X}^t$ is shown in the lower plot in red. Let $\bar{t}_j \coloneqq j$. 
	In the lowest lane $k = 1$, we have $y^\mathcal{A}_{t,1} = 6 \geq 5 = \hat{y}^t_{t,1}$, so server type $y^\mathcal{A}_{t,1}$ will run for at least $\bar{t}_5 = 5$ further time slots (including the current time slot~$t$), i.e., $y^\mathcal{A}_{t,1}$ will be powered down after time slot $t+4$. 
	In lane $k=2$, server type $y^\mathcal{A}_{t-1,2} = 2$ is powered down (because $y^\mathcal{A}_{t-1,2} < \hat{y}^t_{t,2}$) and replaced by $\hat{y}^t_{t,2} = 4$. In lane $k=3$, Algorithm~$\mathcal{A}$ has no active server during time slot $t-1$, so server type $\hat{y}^t_{t,3} = 3$ is powered up.}
	\label{fig:online:const:algo}
\end{figure}
\fi

\textbf{Structure of optimal schedules}
Before we can analyze the competitiveness of algorithm~$\mathcal{A}$, we have to show that an optimal schedule with the desired properties required by line~2 actually exists. First, we will investigate basic properties of optimal schedules. The following lemma shows that in an optimal schedule, a server of type $j$ that runs in lane $k$ does not change the lane while being active.

%First, we prove that a server of type $j$ that runs in lane $k$ does not change the lane while running. 

\begin{lemma}[No lane switching] \label{lemma:online:const:opt:laneswitching}
	Let $\hat{X}$ be an optimal schedule. If $\hat{y}_{t-1,k} = j$ and $\hat{y}_{t,k} \not= j$, then there exists no other lane $k' \not= k$ with $\hat{y}_{t-1,k'} \not= j$ and $\hat{y}_{t,k'} = j$. 
\end{lemma}

\begin{proof}%[Proof of Lemma~\ref{lemma:online:const:opt:laneswitching}]
	Let $\hat{y}_{t-1,k} = j$ and $\hat{y}_{t,k} \not= j$. To get a contradiction, assume that there exists a lane $k' \not= k$ with $\hat{y}_{t-1,k'} \not= j$ and $\hat{y}_{t,k'} = j$. We differ between the cases (1) $k' < k$ and (2) $k' > k$. In case 1, the server type $\hat{y}_{t-1,k'}$ must be greater than $j$, since the server types are sorted. Furthermore, at time slot $t-1$ there are at least $k'$ active servers whose types are greater than $j$, and at time slot $t$ there are at most $k' - 1$ active servers whose types are greater than $j$. Therefore a server of type $j' > j$ is powered down after $t-1$. Let $t' > t$ be the first time slot where $\hat{x}_{t',j} < \hat{x}_{t,j}$. By replacing one server of type $j$ during the time slots $[t:t'-1]$ by $j'$ (i.e., $j'$ is not powered down at $t$, but instead at $t'$), we reduce the operating cost without increasing the switching cost. Therefore, $\hat{X}$ cannot be an optimal schedule.
	
	Case 2 works analogously: we have $k' > k$, so the server type $\hat{y}_{t,k'}$ must be greater than $j$. At time slot $t-1$ there are at most $k-1$ active servers whose types are greater than $j$, and at time slot $t$ there are at least $k$ active servers whose types are greater than $j$. Therefore a server of type $j' > j$ is powered up after $t-1$. Let $t' < t$ be the last time slot where $\hat{x}_{t',j} < \hat{x}_{t,j}$. We replace server type $j$ during $[t'+1:t]$ by $j'$.  The total costs are decreased by this transformation, so $\hat{X}$ cannot be an optimal schedule. Therefore, a lane $k' \not= k$ with $\hat{y}_{t-1,k'} \not= j$ and $\hat{y}_{t,k'} = j$ cannot exist. \qedllncs
\end{proof}

The next lemma shows that in an optimal schedule, a server is only powered up or powered down if the number of jobs is increased or decreased, respectively.
%%LEMMA: server power on/off in optimal schedules
\begin{lemma} \label{lemma:online:const:opt:onoff}
	Let $\hat{X}$ be an optimal schedule. If $\hat{y}_{t-1,k} > 0$ and $\hat{y}_{t,k} = 0$, then $\lambda_{t-1,k} = 1$ and $\lambda_{t,k} = 0$. Analogously, $\hat{y}_{t-1,k} = 0$ and $\hat{y}_{t,k} > 0$ implies $\lambda_{t-1,k} = 0$ and $\lambda_{t,k} = 1$.
\end{lemma}

\begin{proof}%[Proof of Lemma~\ref{lemma:online:const:opt:onoff}]
	Let $\hat{y}_{t-1,k} > 0$ and $\hat{y}_{t,k} = 0$. By Lemma~\ref{lemma:online:const:opt:laneswitching}, we know that a server of type $j \coloneqq \hat{y}_{t-1,k}$ is powered down after time slot $t-1$. There cannot be a job in lane~$k$ at time $t$, because there is no active server in $\hat{X}$, so $\lambda_{t,k} = 0$. Assume that there is no job for the previous time slot, i.e., $\lambda_{t-1,k} = 0$. Then we get a better schedule by powering down the server in lane $k$ one time slot earlier (i.e., after time slot $t-2$), because the operating cost is reduced by $l_j$, so $\hat{X}$ would not be optimal. Therefore, $\lambda_{t-1,k} = 1$ must hold. For $\hat{y}_{t-1,k} = 0$ and $\hat{y}_{t,k} > 0$ the proof works analogously. \qedllncs
\end{proof}

The following lemma shows that in an optimal schedule in a given lane $k$, the server type does not change immediately, i.e., there must be at least one time slot, where no server is running in lane $k$.
\begin{lemma}[No immediate server changes] \label{lemma:online:const:opt:serverchanges}
	Let $\hat{X}$ be an optimal schedule. If $\hat{y}_{t-1,k} > 0$ and $\hat{y}_{t,k} > 0$, then $\hat{y}_{t-1,k} = \hat{y}_{t,k}$ holds.
\end{lemma}

\begin{proof}%[Proof of Lemma~\ref{lemma:online:const:opt:serverchanges}]
	Assume that this lemma does not hold. Let $t$ be the first time slot and $k$ the lowest lane during this time slot where $\hat{y}_{t-1,k} > 0$ and $\hat{y}_{t,k} > 0$, but $\hat{y}_{t-1,k} \not= \hat{y}_{t,k}$. To simplify the notation, let $j \coloneqq \hat{y}_{t-1,k}$ and $j' \coloneqq \hat{y}_{t,k}$. We differ between the cases (1) $j < j'$ and (2) $j > j'$. In case 1, let $t' < t$ be the last time slot where the server type $j$ in lane $k$ was powered up. By replacing server type $j$ by $j'$ during $[t':t-1]$, we reduce the operating cost without increasing the switching cost. If this violates the condition $\hat{x}_{t,j} \leq m_j$, we instead choose the last time slot $t'' \in [t'+1:t-1]$ where $j'$ is powered down. By replacing $j$ with $j'$ during $[t''+1:t-1]$ we reduce the operating cost and save the cost for powering up server type $j'$. It can happen that $j$ has to be powered up one more time, however, the switching cost of $j'$ is smaller than the switching cost of $j$, so the total switching cost is reduced. Case 2 works analogously. We have shown that the total cost can be decreased, so $\hat{X}$ would not be an optimal schedule. Therefore, the lemma must hold. \qedllncs
\end{proof}

%In an optimal schedule $\hat{X}$, a server of type $j$ that runs in lane $k$ does not change the lane while running. Formally, if $\hat{y}_{t-1,k} = j$ and $\hat{y}_{t,k} \not= j$, then there exists no other lane $k' \not= k$ with $\hat{y}_{t-1,k'} \not= j$ and $\hat{y}_{t,k'} = j$. Furthermore, a server is only powered up or powered down if the number of jobs is increased or decreased, respectively. Finally, in a given lane $k$, the server type does not change immediately, i.e., there must be at least one time slot, where no server is running in lane $k$. These properties are proven in the full version of this paper.

Given the optimal schedules $\hat{X}^u$ and $\hat{X}^v$ with $u < v$, we construct a \emph{minimum} schedule $X^{\text{min}(u,v)}$ with $y^{\text{min}(u,v)}_{t,k} \coloneqq \min \{ \hat{y}^u_{t,k}, \hat{y}^v_{t,k} \}$. Furthermore, we construct a \emph{maximum} schedule $X^{\text{max}(u,v)}$ as follows. Let $z_l(t,k)$ be the last time slot $t' < t$ with $\hat{y}^u_{t',k} = \hat{y}^v_{t',k} = 0$ (no active servers in both schedules) and let $z_r(t,k)$ be the first time slot $t' > t$ with $\hat{y}^u_{t',k} = \hat{y}^v_{t',k} = 0$. The schedule $X^{\text{max}(u,v)}$  is defined by 
\begin{equation} \label{eqn:online:const:opt:maxconstruciton}
y^{\text{max}(u,v)}_{t,k} \coloneqq \max_{t' \in [z_l(t,k) + 1 : z_r(t,k) - 1]} \{ \hat{y}^u_{t',k}, \hat{y}^v_{t',k} \} .
\end{equation}
Another way to construct $X^{\text{max}(u,v)}$ is as follows. First, we take the maximum of both schedules (analogously to $X^{\text{min}(u,v)}$). However, this can lead to situations where the server type changes immediately, so the necessary condition for optimal schedules would not be fulfilled. Therefore, we replace the lower server type by the greater one until there are no more immediate server changes. This construction is equivalent to equation~\eqref{eqn:online:const:opt:maxconstruciton}. 

We will see in Lemma~\ref{lemma:online:const:opt:stypes} that the \emph{maximum} schedule is an optimal schedule for $\mathcal{I}^v$ and fulfills the property required by algorithm~$\mathcal{A}$ in line~2, which says that the server type used in lane~$k$ at time~$t$ never decreases when the considered problem instance is expanded. To prove this property, first we have to show that $X^{\text{min}(u,v)}$ and $X^{\text{max}(u,v)}$ are feasible schedules for the problem instances $\mathcal{I}^u$ and $\mathcal{I}^v$, respectively. 

\begin{lemma} \label{lemma:online:const:opt:minmax}
	$X^{\text{min}(u,v)}$ and $X^{\text{max}(u,v)}$ are feasible for $\mathcal{I}^u$ and $\mathcal{I}^v$, respectively.
\end{lemma}

\begin{proof}
	\renewcommand{\labelenumi}{(\alph{enumi})}
	\begin{enumerate}
		\item \textbf{Feasibility of $X^{\text{min}(u,v)}$} \\
		First, we will show that the demand requirements are fulfilled, so for all $k \in [m]$ and $t \in [u]$, there must be an active server in lane $k$ at time $t$, if $\lambda_{t,k} > 0$. Since $\hat{X}^u$ and $\hat{X}^v$ are feasible schedules, $\hat{y}^u_{t,k} \geq \lambda_{t,k}$ and $\hat{y}^v_{t,k} \geq \lambda_{t,k}$ holds for all $t \in [u]$ and $k \in [m]$. Thus, $y^{\text{min}(u,v)}_{t,k} =  \min \{ \hat{y}^u_{t,k}, \hat{y}^v_{t,k} \} \geq \lambda_{t,k}$ holds. 
		
		Second, we have to check if there are not more active servers in $X^{\text{min}(u,v)}$ than available, i.e. $x^{\text{min}(u,v)}_{t,j} \in [m_j]_0$ for all $t \in [u]$ and $j \in [d]$. Assume that this is not the case, so there exists a time slot $t$ and a server type $j$ with $x^{\text{min}(u,v)}_{t,j} > m_j$. Since the server types of $\hat{X}^u$ and $\hat{X}^v$ are sorted, the server types of $X^{\text{min}(u,v)}$ are sorted too. Thus, there must be at least $m_j + 1$ consecutive lanes with $y^{\text{min}(u,v)}_{t,k} = j$. Let $k^+$ be the topmost and $k^-$ be the lowermost lane with $y^{\text{min}(u,v)}_{t,k} = j$. W.l.o.g. let $\hat{y}^u_{t,k^+} = j$ (the case $\hat{y}^v_{t,k^+} = j$ works analogously), so $\hat{y}^v_{t,k^+} \geq j$. It is not possible that $\hat{y}^u_{t,k^-} = j$, because then there would be $m_j + 1$ active servers of type $j$ in $\hat{X}^u$. On the other hand, $\hat{y}^v_{t,k^-} = j$ implies that $\hat{y}^v_{t,k^+} = j$, since the server types are sorted, so there would be at least $m_j + 1$ active servers of type $j$ in $\hat{X}^v$. Thus, our assumption was wrong and $X^{\text{min}(u,v)}$ is a feasible schedule for $\mathcal{I}^u$.
		
		\item \textbf{Feasibility of $X^{\text{max}(u,v)}$} \\
		Consider the schedule $\tilde{X}$ with $\tilde{y}_{t,k} \coloneqq \max\{ \hat{y}^u_{t,k}, \hat{y}^v_{t,k} \}$ (similar to $X^{\text{max}(u,v)}$, but without eliminating immediate server changes). Analogous to part (a), it can be shown that $\tilde{X}$ is a feasible schedule for $\mathcal{I}^v$. Furthermore, we observe that the server types of $\tilde{X}$ are sorted for a given time slot, since the server types of $\hat{X}^u$ and $\hat{X}^v$ are sorted. Taking the maximum preserves this order. 
		
		The schedule $X^{\text{max}(u,v)}$ fulfills the demand requirements of $\mathcal{I}^v$, because $\tilde{y}_{t,k} > 0$ implies $y^{\text{max}(u,v)}_{t,k} > 0$.
		
		Assume that there are more active servers in $X^{\text{max}(u,v)}$ than available, i.e., there exists a time slot $t \in [v]$ and a server type $j \in [d]$ with $x^{\text{max}(u,v)}_{t,j} > m_j$. Let $k^+$ be the topmost lane with $y^{\text{max}(u,v)}_{t,k} = j$. There must be a time slot $t'$ such that $\tilde{y}_{t',k^+} = j$ and $y^{\text{max}(u,v)}_{t'',k^+} = j$ for all $t''$ between $t$ and $t'$ (i.e., $t'' \in [\min\{t, t'\} : \max \{t, t'\}]$), because otherwise $y^ {\text{max}(u,v)}_{t,k^+} = j$ cannot be fulfilled. Let $k^-$ be the lowest lane with $y^{\text{max}(u,v)}_{t,k} = j$. Since the server types in $\tilde{X}$ are sorted and since $\tilde{X}$ is a feasible schedule, $\tilde{y}_{t',k^-} > j$ holds, because $\tilde{y}_{t',k^-} = j$ would imply that $\tilde{X}$ uses server type $j$ in all lanes $k \in [k^- : k^+]$, but ${|[k^- : k^+]|} > m_j$. However, for all $t''$ between $t$ and $t'$ we have $y^{\text{max}(u,v)}_{t'',k^-} > 0$, since there is an active server in the higher lane $k^+$, so $y^{\text{max}(u,v)}_{t,k^-} = y^{\text{max}(u,v)}_{t',k^-} \geq \tilde{y}_{t',k^-} > j$ which is a contradiction to our assumption. Therefore, $X^{\text{max}(u,v)}$ is a feasible schedule for $\mathcal{I}^v$. \qedllncs
	\end{enumerate}
\end{proof}

Now, we are able to show that the maximum schedule is optimal for the problem instance $\mathcal{I}^v$.

%.LEMMA: s-types in new >= s-types in old
\begin{lemma} \label{lemma:online:const:opt:stypes}
	Let $u,v \in [T]$ with $u < v$. $X^{\text{max}(u,v)}$ is optimal for $\mathcal{I}^v$. 
\end{lemma}

%The works roughly as follows (the complete proof can be found in the full paper). First, we prove that the sum of the operating costs of $\hat{X}^u$ and $\hat{X}^v$ is greater than or equal to the sum of the operating cost of $X^{\text{min}(u,v)}$ and $X^{\text{max}(u,v)}$. Each server activation in $X^{\text{min}(u,v)}$ and $X^{\text{max}(u,v)}$ can be mapped to exactly one server activation in $\hat{X}^u$ and $\hat{X}^v$ with the same or a greater server type. Therefore, $C(X^{\text{min}(u,v)}) + C(X^{\text{max}(u,v)}) \leq C(\hat{X}^u) + C(\hat{X}^v)$ holds and by using Lemma~\ref{lemma:online:const:opt:minmax}, it is shown that $X^{\text{max}(u,v)}$ is optimal for $\mathcal{I}^v$.

\begin{proof}
	To simplify the notation, let $X^{\min} \coloneqq X^{\min(u,v)}$ and $X^{\max} \coloneqq X^{\max(u,v)}$.
	Since $\hat{X}^u$ and $\hat{X}^v$ are optimal schedules for $\mathcal{I}^u$ and $\mathcal{I}^v$, respectively, we know from Lemma~\ref{lemma:online:const:opt:minmax} that $C(\hat{X}^u) \leq C(X^{\minuv})$ and $C(\hat{X}^v) \leq C(X^{\maxuv})$. In the following we will show that $C(X^{\minuv}) + C(X^{\maxuv}) \leq C(\hat{X}^u) + C(\hat{X}^v)$ which implies that $X^{\minuv}$ must be an optimal schedule for $\mathcal{I}^u$ and $X^{\maxuv}$ must be an optimal schedule for $\mathcal{I}^v$. First, we compare the operating cost and afterwards the switching cost of the schedules.
	
	The operating costs of $\hat{X}^u$ and $\hat{X}^v$ in lane $k$ at time slot $t$ are 
	\begin{align*}
	l_{\hat{y}^u_{t,k}} + l_{\hat{y}^v_{t,k}} &= l_{\min \{{\hat{y}^u_{t,k}}, \hat{y}^v_{t,k}\}} + l_{\max \{{\hat{y}^u_{t,k}}, \hat{y}^v_{t,k}\} } 
	\geq l_{y^{\minuv}_{t,k}} + l_{y^{\maxuv}_{t,k}} \numberthis\label{eqn:online:const:opt:stypes:lcost}
	\end{align*}
	with $l_0 \coloneqq 0$ (if $y = 0$, then there is no active server, so the operating cost for this time slot is zero). Note that $l_{\min \{{\hat{y}^u_{t,k}}, \hat{y}^v_{t,k}\}} = l_{y^{\minuv}_{t,k}}$ by definition of $X^{\minuv}$ and $l_{\max \{{\hat{y}^u_{t,k}}, \hat{y}^v_{t,k}\} } \geq l_{y^{\maxuv}_{t,k}}$ because $\max \{{\hat{y}^u_{t,k}}, \hat{y}^v_{t,k}\} \leq y^{\maxuv}_{t,k}$. 
	
	Inequality~\eqref{eqn:online:const:opt:stypes:lcost} indicates that the sum of the operating costs of $\hat{X}^{\minuv}$  and $\hat{X}^{\maxuv}$ are smaller than or equal to the sum of the operating costs of $\hat{X}^u$ and $\hat{X}^v$.
	In the following we will show that the same holds for the switching costs.

	Each lane $k$ in the schedule $X^{\maxuv}$ is divided into blocks such that at the beginning of a block a server is powered up and at the end of the block it is powered down. In the following we consider one single block. Let $j$ denote the server type used in that block and let $a$ and $b$ denote the start and end time slot, respectively. Note that in the time slot immediately before the begin and after the end of the block in both $\hat{X}^u$ and $\hat{X}^v$ there is no active server, i.e. $\hat{y}^u_{a-1,k} = \hat{y}^v_{a-1,k} = 0$ and $\hat{y}^u_{b+1,k} = \hat{y}^v_{b+1,k} = 0$. For $t \in [a:b]$, there is always an active server in at least one of the schedules.
	
	For the time interval $[a:b]$ we divide lane $k$ of the schedules $\hat{X}^u$, $\hat{X}^v$ and $X^{\min (u,v)}$ into blocks ${B^w_1, \dots, B^w_{n_w}}$ with $w \in  \{u, v, \minuv\}$ such that at the beginning of the block a server is powered up and at the end of the block it is powered down. Let $j^w_{i}$ denote the server type used in block $B^w_{i}$ with $w \in \{u, v, \minuv\}$ and $i \in [n_w]$. 
	
	In $\hat{X}^u$ or $\hat{X}^v$ (or both) there must be one block $B^{\maxuv}$ with $j^w_{i} = j$ where $w \in \{u, v\}$ (if there are several blocks that fulfill this property, then we choose an arbitrary one). Let $s^{\maxuv}$ denote the start time slot of $B^{\maxuv}$. Let $\mathcal{B}^-$ be the blocks in $\hat{X}^u$ and $\hat{X}^v$ that start before $s^{\maxuv}$ and let $\mathcal{B}^+$ be the blocks that start after $s^{\maxuv}$. Note that $\{\mathcal{B}^-, \mathcal{B}^+, \{B^{\maxuv}\}\}$ is a partition of $\bigcup_{w \in \{u,v\}, i \in [n_w]} B^w_i$. 
	
	\ifincludefigures
	\begin{figure}[t]
	\setlength{\abovecaptionskip}{6pt plus 0pt minus 0pt}
	\setlength{\belowcaptionskip}{0pt plus 0pt minus 0pt}
	
	\centering
	\begin{tikzpicture}
		\pgfmathsetmacro{\xBegin}{0}
		\pgfmathsetmacro{\xStep}{0.6}
		
		\pgfmathsetmacro{\yBeginU}{0}
		\pgfmathsetmacro{\yBeginV}{\yBeginU - 0.9}
		\pgfmathsetmacro{\yBeginMin}{\yBeginV - 1.2}
		\pgfmathsetmacro{\yBeginMax}{\yBeginMin - 0.9}
		\pgfmathsetmacro{\yBeginTime}{\yBeginMax - 0.3}
		\pgfmathsetmacro{\ySize}{0.6}
		
		\pgfmathsetmacro{\tickHalfLength}{0.07}
		
		\tikzstyle{Bminus}=[fill=blue!25!white]
		\tikzstyle{Bplus}=[fill=green!35!white]
		\tikzstyle{Bmax}=[fill=red!30!white]
		\tikzstyle{BminusX}=[fill=blue!10!white]
		\tikzstyle{BplusX}=[fill=green!15!white]
		\tikzstyle{BmaxX}=[fill=red!10!white]
		
		\tikzstyle{connect}=[-stealth, thick]
		
		\draw[Bminus] (\xBegin + 0 * \xStep, \yBeginU) rectangle (\xBegin + 3 * \xStep, \yBeginU + \ySize)
			node[pos=.5] {3};
		\draw[Bminus] (\xBegin + 5 * \xStep, \yBeginU) rectangle (\xBegin + 9 * \xStep, \yBeginU + \ySize)
			node[pos=.5] {7};
		\draw[Bplus] (\xBegin + 12 * \xStep, \yBeginU) rectangle (\xBegin + 15 * \xStep, \yBeginU + \ySize)
			node[pos=.5] {5};
		
		\draw[Bminus] (\xBegin + 0 * \xStep, \yBeginV) rectangle (\xBegin + 1 * \xStep, \yBeginV + \ySize)
			node[pos=.5] {1};
		\draw[Bminus] (\xBegin + 2 * \xStep, \yBeginV) rectangle (\xBegin + 6 * \xStep, \yBeginV + \ySize)
			node[pos=.5] {9};
		\draw[Bmax] (\xBegin + 7 * \xStep, \yBeginV) rectangle 
			%node[below, yshift=-8pt] {$B^{\max}$} 
			(\xBegin + 13 * \xStep, \yBeginV + \ySize)
			node[pos=.5] {13};
			
		\draw[Bplus] (\xBegin + 14 * \xStep, \yBeginV) rectangle (\xBegin + 15 * \xStep, \yBeginV + \ySize)
			node[pos=.5] {2};
		
		\draw[BminusX] (\xBegin + 0 * \xStep, \yBeginMin) rectangle (\xBegin + 1 * \xStep, \yBeginMin + \ySize)
			node[pos=.5] {1};
		\draw[BminusX] (\xBegin + 2 * \xStep, \yBeginMin) rectangle (\xBegin + 3* \xStep, \yBeginMin + \ySize)
			node[pos=.5] {3};
		\draw[BminusX] (\xBegin + 5 * \xStep, \yBeginMin) rectangle (\xBegin + 6 * \xStep, \yBeginMin + \ySize)
			node[pos=.5] {7};
		\draw[BminusX] (\xBegin + 7 * \xStep, \yBeginMin) rectangle (\xBegin + 9 * \xStep, \yBeginMin + \ySize)
			node[pos=.5] {7};
		\draw[BplusX] (\xBegin + 12 * \xStep, \yBeginMin) rectangle (\xBegin + 13 * \xStep, \yBeginMin + \ySize)
			node[pos=.5] {5};
		\draw[BplusX] (\xBegin + 14 * \xStep, \yBeginMin) rectangle (\xBegin + 15 * \xStep, \yBeginMin + \ySize)
			node[pos=.5] {2};
		
		\draw[BmaxX] (\xBegin + 0 * \xStep, \yBeginMax) rectangle (\xBegin + 15 * \xStep, \yBeginMax + \ySize)
			node[pos=.5] {13};
		
		\draw[connect, bend left=15] (\xBegin + 1 * \xStep, \yBeginMin + \ySize) to (\xBegin + 1 * \xStep, \yBeginV);
		\draw[connect, bend left=15] (\xBegin + 3 * \xStep, \yBeginMin + \ySize) to (\xBegin + 3 * \xStep, \yBeginU);
		\draw[connect, bend left=15] (\xBegin + 6 * \xStep, \yBeginMin + \ySize) to (\xBegin + 6 * \xStep, \yBeginV);
		\draw[connect, bend left=15] (\xBegin + 9 * \xStep, \yBeginMin + \ySize) to (\xBegin + 9 * \xStep, \yBeginU);
		
		\draw[connect, bend right=15] (\xBegin + 12 * \xStep, \yBeginMin + \ySize) to (\xBegin + 12 * \xStep, \yBeginU);
		\draw[connect, bend right=15] (\xBegin + 14 * \xStep, \yBeginMin + \ySize) to (\xBegin + 14 * \xStep, \yBeginV);
		
		\draw[connect] (\xBegin + 10.5 * \xStep, \yBeginMax + \ySize) to (\xBegin + 10.5 * \xStep, \yBeginV);
		
		\node[left] at (\xBegin, \yBeginU + \ySize / 2) {$\hat{y}^u_{t,k}$};
		\node[left] at (\xBegin, \yBeginV + \ySize / 2) {$\hat{y}^v_{t,k}$};
		\node[left] at (\xBegin, \yBeginMin + \ySize / 2) {${y}^{\minuv}_{t,k}$};
		\node[left] at (\xBegin, \yBeginMax + \ySize / 2) {${y}^{\maxuv}_{t,k}$};
		
		\draw[-stealth] (\xBegin - \xStep, \yBeginTime) to (\xBegin + 16 * \xStep, \yBeginTime);
		\foreach \x in {0,...,15 } {
			\draw (\xBegin + \x * \xStep, \yBeginTime - \tickHalfLength) to (\xBegin + \x * \xStep, \yBeginTime + \tickHalfLength);
		}
		\node[] at (\xBegin + 0.5 * \xStep, \yBeginTime - \tickHalfLength * 3) {$a$};
		\node[] at (\xBegin + 14.5* \xStep, \yBeginTime - \tickHalfLength * 3) {$b$};
		%\node[] at (\xBegin + 7.5* \xStep, \yBeginTime - \tickHalfLength * 3) {$s_{\max}$};
		
	\end{tikzpicture}
	\caption{{\normalfont (figure is colored)} Visualization of the proof of Lemma~\ref{lemma:online:const:opt:stypes}. \normalfont The number inside each block refers to the used server type. The blocks of the sets $\mathcal{B}^-$ and $\mathcal{B}^+$ are marked in blue and green, respectively. Block $B^{\max}$, which contains the largest server type, is drawn in red. Note that the third block $B^{\minuv}_3$ in $X^{\minuv}$ is mapped to the second block $B^{v}_2$ in $\hat{X}^v$, but it uses the server type $j^{\minuv}_3 = \min \{j^u_2, j^v_2\} = j^u_2 = 7$ instead of $j^{v}_2 = 9$. However, since $\beta_7 < \beta_9$, the switching cost of $B^{\minuv}_3$ is smaller than the switching cost of the assigned block $B^{v}_2$.} 
	\label{fig:online:const:opt:stypes}
\end{figure}
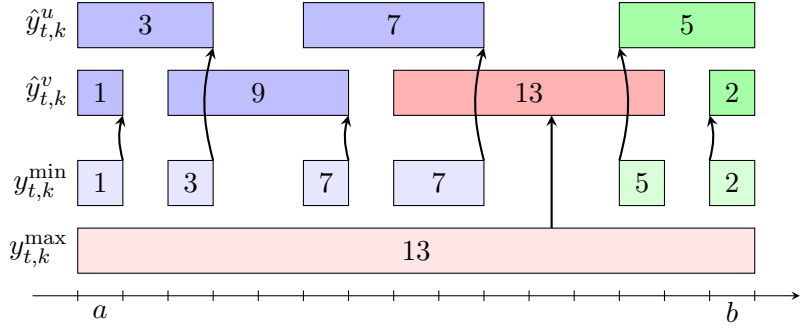
	\fi
	
	Each block $B^{\minuv}_i$ which starts before $s^{\maxuv}$ is mapped to the block in $\mathcal{B}^-$ which has the same \emph{end} time slot. There must be a block $B^{\minuv}_i$ which starts at $s^{\maxuv}$. This block is mapped to the last block in $\mathcal{B}^-$ (which cannot end before $s^{\maxuv}$, so it was not mapped yet).
	Each block $B^{\minuv}_i$ which starts after $s^{\maxuv}$ is mapped to the block in $\mathcal{B}^+$ which has the same \emph{start} time slot. The mapping procedure is visualized in Figure~\ref{fig:online:const:opt:stypes}. It ensures that all blocks $B^{\minuv}_i$ with $i \in [n_w]$ are mapped to a block of $\hat{X}^u$ or $\hat{X}^v$, but not to the block $B^{\maxuv}$. Since $X^{\minuv}$ uses the smaller server type of $\hat{X}^u$ and $\hat{X}^v$, the switching cost of $B^{\minuv}_i$ is smaller than or equal to the switching cost of the mapped block $B^w_i$ with $w \in \{u,v\}$.
	
	Let $\beta(B)$ denote the switching cost of block $B$. The switching costs of $\hat{X}^u$ and $\hat{X}^v$ in lane $k$ during the time interval $[a:b]$ are equal to 
	$%\begin{equation*}
	\beta(B^{\maxuv}) + \sum_{B \in \mathcal{B}^- \cap \mathcal{B}^+} \beta(B)
	$. %\end{equation*}
	The switching cost of $X^{\minuv}$ in lane $k$ during $[a:b]$ is at most $\sum_{B \in \mathcal{B}^- \cap \mathcal{B}^+} \beta(B)$ and the switching cost of $X^{\maxuv}$ is exactly $\beta(B^{\maxuv})$, because $X^{\maxuv}$ only consists of one single block. By using this result for all blocks of $X^{\maxuv}$ and with equation~\eqref{eqn:online:const:opt:stypes:lcost}, we get  $C(X^{\minuv}) + C(X^{\maxuv}) \leq C(\hat{X}^u) + C(\hat{X}^v)$ which implies that $X^{\maxuv}$ must be an optimal schedule for $\mathcal{I}^v$. \qedllncs
\end{proof}

\vspace{8pt}
\textbf{Feasibility}
In the following, let $\{\hat{X}^1, \dots, \hat{X}^T\}$ be optimal schedules that fulfill the inequality $\hat{y}^t_{t',k} \geq \hat{y}^{t-1}_{t',k}$ for all $t, t' \in [T]$ and $k \in [m]$ as required by algorithm~$\mathcal{A}$. Lemma~\ref{lemma:online:const:opt:stypes} ensures that such a schedule sequence exists (and also shows how to construct it). 
Before we can prove that algorithm~$\mathcal{A}$ is $2d$-competitive, we have to show that the computed schedule $X^\mathcal{A}$ is feasible.

The following lemma shows that the running times $\bar{t}_j$ are sorted in ascending order, i.e., $\bar{t}_1 \leq \dots \leq \bar{t}_d$. In other words, the higher the server type is, the longer it stays in the active state. 

\begin{lemma} \label{lemma:online:const:alg:bart}
	For $j < j'$, $\bar{t}_j \leq \bar{t}_{j'}$ holds.
\end{lemma}

\begin{proof} 
	Since $j < j'$, we have $l_j > l_{j'}$ and $\beta_j < \beta_{j'}$, so $\bar{t}_j = \left\lfloor \beta_j / l_j \right\rfloor \leq \left\lfloor \beta_{j'} / l_{j'} \right\rfloor = \bar{t}_{j'}$. \qedllncs
\end{proof}

In an optimal schedule~$\hat{X}^t$, the values $\hat{y}^t_{t',1}, \dots, \hat{y}^t_{t',m}$ are sorted in descending order by definition. This also holds for the schedule calculated by our algorithm.

\begin{lemma} \label{lemma:online:const:alg:ysorting}
	For all time slots $t \in [T]$, the values $y^{\mathcal{A}}_{t,1}, \dots, y^{\mathcal{A}}_{t,m}$ are sorted in descending order, i.e., $y^{\mathcal{A}}_{t,k} \geq y^{\mathcal{A}}_{t,k'}$ for $k < k'$. 
\end{lemma}

\begin{proof}%[Proof of Lemma~\ref{lemma:online:const:alg:ysorting}]
	Assume that Lemma~\ref{lemma:online:const:alg:ysorting} does not hold. Let $t$ be the first time slot with $y^{\mathcal{A}}_{t,k} < y^{\mathcal{A}}_{t,k'}$. 
	%To simplify the notation, let $j \coloneqq y^{\mathcal{A}}_{t,k}$ and $j' \coloneqq y^{\mathcal{A}}_{t,k'}$. 
	If $y^{\mathcal{A}}_{t,k'}$ is powered up at time $t$, then $\hat{y}^t_{t,k'} = y^{\mathcal{A}}_{t,k'}$ holds. By the definition of algorithm~$\mathcal{A}$, the server types used during a given time slot are greater than or equal to the server types used by $\hat{X}^t$, so $y^{\mathcal{A}}_{t,k} \geq \hat{y}^t_{t,k}$. The server types in $\hat{X}^t$ are sorted, so we get $y^{\mathcal{A}}_{t,k} \geq \hat{y}^t_{t,k} \geq \hat{y}^t_{t,k'} = y^{\mathcal{A}}_{t,k'}$ which contradicts our assumption.
	
	If $y^{\mathcal{A}}_{t,k'}$ is already running at time $t$, we consider the time slot~$t' < t$ when the value of $e_{k'}$ has changed for the last time. Formally, let $t' < t$ be the last time slot such that $t' + \bar{t}_{\hat{y}^{t'}_{t',k'}} > t$. We have $\hat{y}^{t'}_{t',k} \geq \hat{y}^{t'}_{t',k'}$, so by Lemma~\ref{lemma:online:const:alg:bart}, $y^\mathcal{A}_{t',k}$ runs at least as long as $y^\mathcal{A}_{t',k'}$. Therefore, the fact $y^\mathcal{A}_{t',k} \geq y^\mathcal{A}_{t',k'}$ implies $y^\mathcal{A}_{t,k} \geq y^\mathcal{A}_{t,k'}$ which is a contradiction to our assumption.  \qedllncs
\end{proof}

%The proof uses the fact that the running times $\bar{t}_j$ are sorted in ascending order, i.e., $\bar{t}_1 \leq \dots \leq \bar{t}_d$, because $l_1 > \dots > l_d$ and $\beta_1 < \dots < \beta_d$. In other words, the higher the server type is, the longer it stays in the active state. See the full paper for more details. 

%By means of Lemma~\ref{lemma:online:const:alg:ysorting}, we are able to prove the feasibility of $X^\mathcal{A}$.
Now, we are able to prove the feasibility of $X^\mathcal{A}$.

\begin{lemma} \label{lemma:online:const:alg:feasible}
	The schedule $X^\mathcal{A}$ is feasible.
\end{lemma}

\begin{proof}%[Proof of Lemma~\ref{lemma:online:const:alg:feasible}]
	A schedule is feasible, if (1) there are enough active servers to handle the incoming jobs (i.e., $\sum_{j=1}^{d} x^\mathcal{A}_{t,j} \geq \lambda_t$) and (2) there are not more active servers than available (i.e., $x^\mathcal{A}_{t,j} \in [m_j]_0$).
	
	\renewcommand{\labelenumi}{(\arabic{enumi})}
	\begin{enumerate}
		\item By the definition of algorithm~$\mathcal{A}$, the server types used during a given time slot are greater than or equal to the server types used by $\hat{X}^t$, so there are at least as many active servers as in $\hat{X}^t$. Therefore, $\sum_{j=1}^{d} x^\mathcal{A}_{t,j} \geq \sum_{j=1}^{d} \hat{x}^t_{t,j} \geq \lambda_t$ holds for all $t \in [T]$. 
		
		\item Assume that there exist $t \in [T]$ and $j \in [d]$ such that  $x^\mathcal{A}_{t,j} > m_j$. Let $t$ be the first time slot where algorithm~$\mathcal{A}$ wants to use server type $j$ in lane $k$, although it is used already $m_j$ times in the lower lanes during the same time slot. Let $K$ be the set of lanes where $j$ is already used, i.e., $y^{\mathcal{A}}_{t,k'} = j$ for all $k' \in K \subseteq [k-1]$. We differ between case 1 where $y^{\mathcal{A}}_{t,k}$ is set in line 5 and case 2 where  $y^{\mathcal{A}}_{t,k}$ is set in line 8. 
		
		In the first case ($y^{\mathcal{A}}_{t,k}$ is set in line 5), we know that $\hat{X}^t$ uses $j$ in lane $k$. Since the server types of $\hat{X}^t$ are sorted, the server types of $\hat{X}^t$ in the lower lanes cannot be smaller than $k$. Formally, we have $\hat{y}^t_{t,k'} \geq j$ for all $k' \in [k]$. In the lanes where $\mathcal{A}$ uses server type~$j$, the optimal schedule $\hat{X}^t$ cannot use a greater server type. Thus, there are exactly $m_j$ lanes below lane $k$ where $\hat{y}^t_{t,k'} = j$ holds, so $\hat{X}^t$ cannot use $j$ in lane $k$. 
		
		In the second case ($y^{\mathcal{A}}_{t,k}$ is set in line 8), we know that $y^{\mathcal{A}}_{t-1,k} = j$, but $x^\mathcal{A}_{t-1,j} \leq m_j$, so there must be a lane $k' \in K$ with $y^{\mathcal{A}}_{t-1,k'} > j$ by Lemma~\ref{lemma:online:const:alg:ysorting}. We consider the time slot $t'$ when the value of $e_k$ has changed for the last time. Formally, let $t' < t$ be the last time slot such that $t' + \bar{t}_{\hat{y}^{t'}_{t',k}} > t$. We know that $y^{\mathcal{A}}_{t',k'} = y^{\mathcal{A}}_{t-1,k'}$, because $y^{\mathcal{A}}_{t-1,k'} > j$ cannot be powered up during $[t':t-1]$ and powered down at $t$, as $y^{\mathcal{A}}_{t-1,k'} > j \geq \hat{y}^{t'}_{t',k}$ implies $\bar{t}_{y^{\mathcal{A}}_{t-1,k'}} \geq \bar{t}_{\hat{y}^{t'}_{t',k}}$. Since $\hat{y}^{t'}_{t',k} \leq \hat{y}^{t'}_{t',k'}$ holds, the runtime of $y^{\mathcal{A}}_{t',k}$ in lane~$k'$ was extended at time slot $t'$, so it still runs during time slot~$t$. This is a contradiction to $y^{\mathcal{A}}_{t,k'} = j$. \qedllncs %\qedhere
	\end{enumerate}%
\end{proof}

%\begin{proof}[Proof idea]
%	A schedule is feasible, if (1) there are enough active servers to handle the incoming jobs (i.e., $\sum_{j=1}^{d} x^\mathcal{A}_{t,j} \geq \lambda_t$) and (2) there are not more active servers than available (i.e., $x^\mathcal{A}_{t,j} \leq m_j$). The first property directly follows from the definition of algorithm~$\mathcal{A}$, since $\sum_{j=1}^{d} x^\mathcal{A}_{t,j} \geq \sum_{j=1}^{d} \hat{x}^t_{t,j} \geq \lambda_t$. %For the second property, we differ between the cases if $y^\mathcal{A}_{t,k}$ is set in line 5 or line 8. 
%	Lemma~\ref{lemma:online:const:alg:ysorting} is used to prove that $x^\mathcal{A}_{t,j} \leq m_j$ is always fulfilled after setting $y^\mathcal{A}_{t,k}$ in line 5 or 8.
%	The complete proof is presented in the full paper. \qedllncs
%\end{proof}

\textbf{Competitiveness}
To show the competitiveness of $\mathcal{A}$, we divide the schedule $X^\mathcal{A}$ into blocks $A_{t,k}$ with $t \in [T]$ and $k \in [m]$. Each block $A_{t,k}$ is described by its creation time $t$, its start time $s_{t,k}$, its end time $e_{t,k}$, the used server type $j_{t,k}$ and the corresponding lane~$k$. The start time is the time slot when $j_{t,k}$ is powered up and the end time is the first time slot, when $j_{t,k}$ is inactive, i.e., during the time interval $[s_{t,k} : e_{t,k} - 1]$ the server of type $j_{t,k}$  is in the active state.

There are two types of blocks: \emph{new} blocks and \emph{extended} blocks. A \emph{new} block starts when a new server is powered up, i.e., lines 5 and 6 of algorithm~$\mathcal{A}$ are executed because $y^{\mathcal{A}}_{t-1,k} < \hat{y}^t_{t,k}$ or $t \geq e_k \land y^{\mathcal{A}}_{t-1,k} > \hat{y}^t_{t,k} \land \hat{y}^t_{t,k} > 0$ (in words: the previous block ends and $\hat{X}^t$ has an active server in lane $k$, but the server type is smaller than the server type used by $\mathcal{A}$ in the previous time slot). It ends after $\bar{t}_{y^{\mathcal{A}}_{t,k}}$ time slots. Thus $s_{t,k} \coloneqq t$ and $e_{t,k} \coloneqq t + \bar{t}_{y^{\mathcal{A}}_{t,k}}$ (i.e.,  $e_{t,k}$ equals $e_k$ after executing line 6). 

An \emph{extended} block is created when the running time of a server is extended, i.e., 
the value of $e_k$ is updated, but the server type remains the same (that is $y^{\mathcal{A}}_{t-1,k} = y^{\mathcal{A}}_{t,k}$). We have $e_{t,k} \coloneqq t + \bar{t}_{\hat{y}^t_{t,k}}$ (i.e., the value of $e_k$ after executing line~9 or~6) and $s_{t,k} \coloneqq e_{t',k}$, where $A_{t',k}$ is the previous block in the same lane.
Note that an \emph{extended} block can be created not only in line~9, but also in line~6, if $t = e_k$ and $y^{\mathcal{A}}_{t-1,k} = \hat{y}^t_{t,k}$.
%Note that an \emph{extended} block can be created not only in line~8 and~9, but also in line~5 and~6, if $t = e_k$ and $y^{\mathcal{A}}_{t-1,k} = \hat{y}^t_{t,k}$.
%ine~8 and~9 are executed because $y^{\mathcal{A}}_{t-1,k} \geq \hat{y}^t_{t,k}$ or line~5 and~6 are executed because $t = e_k$ and $y^{\mathcal{A}}_{t-1,k} = \hat{y}^t_{t,k}$ (i.e., server type $y^{\mathcal{A}}_{t-1,k}$ is already running and does not need to be powered up). 
If line~8 and~9 are executed, but the value of $e_k$ does not change (because $t + \bar{t}_{\hat{y}^t_{t,k}}$ is smaller than or equal to the previous value of $e_k$), then the block $A_{t,k}$ does not exist. Figure~\ref{fig:online:const:atk} visualizes the definition of~$A_{t,k}$.

\ifincludefigures
\begin{figure}[t]
	\setlength{\abovecaptionskip}{6pt plus 0pt minus 0pt}
	\setlength{\belowcaptionskip}{0pt plus 0pt minus 0pt}
	
	\newcommand{\xx}{\mkern-2mu}
	\small
	\centering
	\begin{tikzpicture}

	\def\yopt{{0,1,2,1,2,1,1,0,2,0,0,0,3,2,1,2,0,0,0}}
	\def\yalg{{0,1,2,2,2,2,2,2,2,2,2,0,3,3,3,3,3,3,0}}
	
	\pgfmathsetmacro{\xBegin}{0}
	\pgfmathsetmacro{\xStep}{0.545}
	\pgfmathsetmacro{\xLeft}{-0}
	\pgfmathsetmacro{\xLength}{19}
	\pgfmathsetmacro{\tMax}{\xLength - 1}
	
	\pgfmathsetmacro{\yBeginT}{-3.2}
	\pgfmathsetmacro{\yBeginOpt}{-0.8}
	\pgfmathsetmacro{\yBeginAlg}{-1.6}
	\pgfmathsetmacro{\yBeginAtk}{-2.4}
	\pgfmathsetmacro{\yHeight}{\xStep}
	\pgfmathsetmacro{\axisExtend}{0.7}
	\pgfmathsetmacro{\yHeightCorrection}{-0.025}
	
	\pgfmathsetmacro{\tickHalfLength}{0.07}
	
	\tikzstyle{arrow}=[thick, -stealth]
	\tikzstyle{circlenode}=[draw, circle, thick, inner sep=2pt]
	\tikzstyle{blocknew}=[thick, fill=green!30!white]
	\tikzstyle{blockextended}=[thick, fill=blue!20!white]%[fill=red!20!white]
	\tikzstyle{blocktext}=[font=\scriptsize]%[fill=red!20!white]

	\node[left] at (\xBegin + \xLeft, \yBeginOpt) {$\hat{y}^t_{t,k} = $};
	\node[left] at (\xBegin + \xLeft, \yBeginAlg) {$y^{\mathcal{A}}_{t,k} =$};
	
	\foreach \x in {0,...,\tMax} {
		\pgfmathsetmacro{\yoptx}{\yopt[\x]}
		\pgfmathsetmacro{\yalgx}{\yalg[\x]}
		%\node[] at (\xBegin + \x * \xStep + 0.5 * \xStep, \yBeginT) {$\x$}; 
		\node[] at (\xBegin + \x * \xStep + 0.5 * \xStep, \yBeginOpt) {\yoptx}; 
		\node[] at (\xBegin + \x * \xStep + 0.5 * \xStep, \yBeginAlg) {\yalgx}; 
	}
	
	\draw[arrow] (\xBegin, \yBeginT) to (\xBegin + \xStep * \xLength +  \axisExtend * \xStep, \yBeginT)
	node[below] {$t$};
	\foreach \x in {1,...,\xLength } {
		\draw (\xBegin + \x * \xStep, \yBeginT - \tickHalfLength) to (\xBegin + \x * \xStep, \yBeginT + \tickHalfLength);
		\pgfmathsetmacro{\t}{int(\x - 1)}
		\node[below] at (\xBegin + \x * \xStep - 0.5 * \xStep, \yBeginT) {\t};
	}
	
	\node[circlenode] (a1) at (\xBegin + 1.5 * \xStep, \yBeginOpt) {$\phantom{0}$};
	\draw[arrow, bend right=15] (a1) to (\xBegin + 1.5 * \xStep, \yBeginAtk);
	
	\node[circlenode] (a2) at (\xBegin + 2.5 * \xStep, \yBeginOpt) {$\phantom{0}$};
	\draw[arrow, bend right=15] (a2) to (\xBegin + 2.5 * \xStep, \yBeginAtk);
	
	\node[circlenode] (a3) at (\xBegin + 4.5 * \xStep, \yBeginOpt) {$\phantom{0}$};
	\draw[arrow, bend left=15] (a3) to (\xBegin + 5.5 * \xStep, \yBeginAtk);
	
	\node[circlenode] (a4) at (\xBegin + 6.5 * \xStep, \yBeginOpt) {$\phantom{0}$};
	\draw[arrow, bend left=15] (a4) to (\xBegin + 7.5 * \xStep, \yBeginAtk);
	
	\node[circlenode] (a5) at (\xBegin + 8.5 * \xStep, \yBeginOpt) {$\phantom{0}$};
	\draw[arrow, bend left=15] (a5) to (\xBegin + 8.5 * \xStep, \yBeginAtk);
	
	\node[circlenode] (a6) at (\xBegin + 12.5 * \xStep, \yBeginOpt) {$\phantom{0}$};
	\draw[arrow, bend right=15] (a6) to (\xBegin + 12.5 * \xStep, \yBeginAtk);
	
	\node[circlenode] (a7) at (\xBegin + 15.5 * \xStep, \yBeginOpt) {$\phantom{0}$};
	\draw[arrow, bend left=15] (a7) to (\xBegin + 17.5 * \xStep, \yBeginAtk);
	
	\draw[blocknew] (\xBegin + 1 * \xStep, \yBeginAtk) rectangle
		(\xBegin + 2 * \xStep, \yBeginAtk - \yHeight);
	\draw[blocknew] (\xBegin + 2* \xStep, \yBeginAtk) rectangle
		(\xBegin + 5 * \xStep, \yBeginAtk - \yHeight);
	\draw[blockextended] (\xBegin + 5* \xStep, \yBeginAtk) rectangle
		(\xBegin + 7 * \xStep, \yBeginAtk - \yHeight);
	\draw[blockextended] (\xBegin + 7* \xStep, \yBeginAtk) rectangle
		(\xBegin + 8 * \xStep, \yBeginAtk - \yHeight);
	\draw[blockextended] (\xBegin + 8* \xStep, \yBeginAtk) rectangle
		(\xBegin + 11* \xStep, \yBeginAtk - \yHeight);
	\draw[blocknew] (\xBegin + 12* \xStep, \yBeginAtk) rectangle
		(\xBegin + 17 * \xStep, \yBeginAtk - \yHeight);
	\draw[blockextended] (\xBegin + 17* \xStep, \yBeginAtk) rectangle
		(\xBegin + 18* \xStep, \yBeginAtk - \yHeight);
		
	\draw[dashed,-|] (\xBegin + 2 * \xStep, \yBeginAtk -  0.5 * \yHeight) --
		(\xBegin + 3 * \xStep, \yBeginAtk - 0.5 * \yHeight );
		
	\node[blocktext] at (\xBegin + 1.5 * \xStep, \yBeginAtk - 0.5 * \yHeight + \yHeightCorrection) {$A_{1\xx,\xx k}$};
	\node[blocktext] at (\xBegin + 3.5 * \xStep, \yBeginAtk - 0.5 * \yHeight + \yHeightCorrection) {$A_{2\xx,\xx k}$};
	\node[blocktext] at (\xBegin + 5.999 * \xStep, \yBeginAtk - 0.5 * \yHeight + \yHeightCorrection) {$A_{4\xx,\xx k}$};
	\node[blocktext] at (\xBegin + 7.5 * \xStep, \yBeginAtk - 0.5 * \yHeight + \yHeightCorrection) {$A_{6\xx,\xx k}$};
	\node[blocktext] at (\xBegin + 9.5 * \xStep, \yBeginAtk - 0.5 * \yHeight + \yHeightCorrection) {$A_{8\xx,\xx k}$};
	\node[blocktext] at (\xBegin + 14.5 * \xStep, \yBeginAtk - 0.5 * \yHeight + \yHeightCorrection) {$A_{12\xx,\xx k}$};
	\node[blocktext] at (\xBegin + 17.5 * \xStep, \yBeginAtk - 0.5 * \yHeight + \yHeightCorrection) {$A_{\xx1\xx5\xx,\xx k}$};
	
	\end{tikzpicture}
	\caption{{\normalfont (figure is colored)} Visualization of the definition of the blocks $A_{t,k}$ for one specific lane~$k$. The first line shows the values of $\hat{y}^t_{t,k}$ for $t \in [0:18]$ and the second line the resulting schedule of algorithm~$\mathcal{A}$. In this example, we have $(\bar{t}_1, \bar{t}_2, \bar{t}_3) = (2,3,5)$. The blocks $A_{t,k}$ are printed as rectangles that show the start and end time $s_{t,k}$ and $e_{t,k}$, e.g., $s_{4,k} = 5$ and $e_{4,k} = 7$ (note that the server is in the active state during $[s_{t,k} : e_{t,k}-1]$). The dashed line after $A_{1,k}$ indicates that $e_{1,k} = 3$, since the block is interrupted by $A_{2,k}$. \emph{New} blocks are drawn in green and \emph{extended} blocks are drawn in blue.  The arrows indicate the creation time $t$ of a block. The used server type $j_{t,k}$ is equal to $y^\mathcal{A}_{t,k}$. Blocks that are not printed do not exist, e.g., $A_{3,k}$ does not exist, because $t + \bar{t}_{\hat{y}^t_{t,k}} = 3 + \bar{t}_{\hat{y}^3_{3,k}} = 3 + \bar{t}_1 = 3 + 2 = 5 \leq e_{2,k} = 5$, so $e_k$ is not updated at $t=3$.}
		%Note that the values $\hat{y}^t_{t,k}$ do not represent an optimal schedule, but the last state of each optimal schedule $\hat{X}^t$.  }
	\label{fig:online:const:atk}
\end{figure}
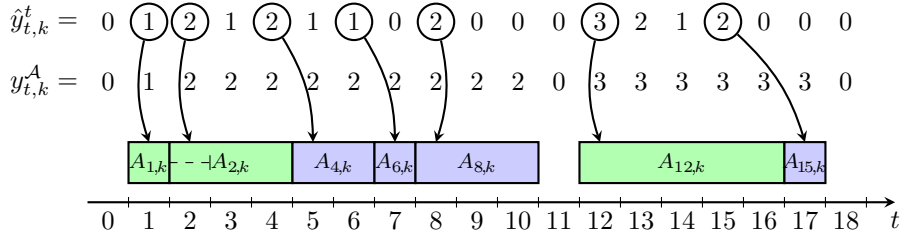
\fi

Let $d_{t,k} \coloneqq e_{t,k} - s_{t,k}$ be the duration of the block $A_{t,k}$ and let $C(A_{t,k})$ be the cost caused by $A_{t,k}$ if the block $A_{t,k}$ exists or 0 otherwise. The next lemma describes how the cost of a block can be estimated.

\begin{lemma} \label{lemma:online:const:alg:costai}
	The cost of the block $A_{t,k}$ is upper bounded by
	\begin{equation}
	C(A_{t,k}) \leq \begin{cases}
	2 \beta_{j_{t,k}} & \text{if $A_{t,k}$ is a \emph{new} block} \\
	l_{j_{t,k}}  d_{t,k} & \text{if $A_{t,k}$ is an \emph{extended} block}
	\end{cases}
	\end{equation}
\end{lemma}

\begin{proof}
	If $A_{t,k}$ is a \emph{new} block, its length is $\bar{t}_{j}$ with $j = j_{t,k}$. Therefore the total cost of $A_{t,k}$ is $\beta_j + l_j \bar{t}_j = \beta_j + l_j \left\lfloor {\beta_j}/{l_j} \right\rfloor \leq 2 \beta_j$.
	If $A_{t,k}$ is an \emph{extended} block, then the server $j = j_{t,k}$ is already running, so there is no switching cost and $C(A_{t,k}) = l_j d_{t,k}$. \qedllncs
\end{proof}

To show the competitiveness of algorithm~$\mathcal{A}$, we introduce another variable that will be used in Lemmas~\ref{lemma:online:const:opt:scaledcost} and~\ref{lemma:online:const:opt:mainlemma}. Let 
\begin{equation*}
\tilde{y}^u_{t,k} \coloneqq \max_{t' \in [t:u]} \hat{y}^{t'}_{t',k} 
\end{equation*}
be the largest server type used in lane $k$ by the schedule $\hat{X}^{t'}$ at time slot $t'$ for $t' \in [t:u]$. The next lemma shows that $\tilde{y}^u_{t,k}$ is monotonically decreasing with respect to $t$ as well as $k$ and increasing with respect to $u$. 

\begin{lemma} \label{lemma:online:const:opt:ytilde}
	Let $u' \geq u$, $t' \leq t$ and $k' \leq k$. It is $\tilde{y}^u_{t,k} \leq \tilde{y}^{u'}_{t',k'}$.
\end{lemma}

\begin{proof}
	We analyze the special cases where two inequalities are fulfilled with equality.%
	\begin{enumerate}
		\item If $t' = t$ and $k' = k$ holds, then we have
		\begin{equation} \label{eqn:online:const:opt:ytilde:u}
		\tilde{y}^{u'}_{t,k} = \max_{t'' \in [t:u']} \hat{y}^{t''}_{t'',k} \leq \max_{t'' \in [t:u]} \hat{y}^{t''}_{t'',k} = \tilde{y}^{u}_{t,k}.
		\end{equation}
		
		\item Let $u' = u$ and $k' = k$. By using the definition of $\tilde{y}^u_{t,k}$, we get
		\begin{equation}\label{eqn:online:const:opt:ytilde:t}
		\tilde{y}^u_{t,k} = \max_{t'' \in [t:u]} \hat{y}^{t''}_{t'',k} \leq \max_{t'' \in [t':u]} \hat{y}^{t''}_{t'',k} = \tilde{y}^u_{t',k}
		\end{equation}
		since $[t:u] \subset [t':u]$.
		
		\item If $u' = u$ and $t' = t$, then we can use the fact that the server types of $\hat{X}^{t'}$ at time slot $t'$ are sorted, i.e., $\hat{y}^{t''}_{t'',k} \leq  \hat{y}^{t''}_{t'',k'}$ holds for all $t'' \in [T]$ and $k' \leq k$.  Therefore,
		\begin{equation}\label{eqn:online:const:opt:ytilde:k}
		\tilde{y}^u_{t,k} = \max_{t'' \in [t:u]} \hat{y}^{t''}_{t'',k} \leq \max_{t'' \in [t:u]} \hat{y}^{t''}_{t'',k'} = \tilde{y}^u_{t,k'}.
		\end{equation}
	\end{enumerate}
	
	By using equations~\eqref{eqn:online:const:opt:ytilde:u}, ~\eqref{eqn:online:const:opt:ytilde:t} and~\eqref{eqn:online:const:opt:ytilde:k}, we get $\tilde{y}^u_{t,k} \leq \tilde{y}^{u'}_{t',k} \leq \tilde{y}^{u'}_{t',k} \leq \tilde{y}^{u'}_{t',k'}$. \qedllncs
\end{proof}

The cost of schedule $X$ in lane $k$ during time slot $t$ is denoted by
\begin{equation}\label{eqn:online:const:ctkdef}
C_{t,k}(X) \coloneqq 
\begin{cases} 
l_{y_{t,k}} + \beta_{y_{t,k}} & \text{if $y_{t-1,k} \not= y_{t,k} > 0$} \\
l_{y_{t,k}} & \text{if $y_{t-1,k} = y_{t,k} > 0$} \\
0 & \text{otherwise.}
\end{cases}
\end{equation} 
The total cost of $X$ can be written as 
%\begin{equation*}
$C(X) = \sum_{t=1}^{T} \sum_{k=1}^{m} C_{t,k}(X)$.
%\end{equation*}
%
The technical lemma below will be needed for our induction proof in  Theorem~\ref{theorem:online:const}. Given the optimal schedules $\hat{X}^u$ and $\hat{X}^v$ with $u < v$, the inequality $\sum_{k=1}^m \sum_{t=1}^{u} C_{t,k}(\hat{X}^u) \leq \sum_{k=1}^m \sum_{t=1}^{u} C_{t,k}(\hat{X}^v)$ is obviously fulfilled (because $\hat{X}^u$ is an optimal schedule for $\mathcal{I}^u$, so $\hat{X}^v$ cannot be better). The lemma below shows that this inequality still holds if the cost $C_{t,k}(\cdot)$ is scaled by $\tilde{y}^u_{t,k}$.

%.LEMMA: scaled cost
\begin{lemma} \label{lemma:online:const:opt:scaledcost}
	Let $u, v \in [T]$ with $u < v$. It holds that
	\begin{equation} \label{eqn:online:const:opt:scaledcost}
	\sum_{k=1}^m \sum_{t=1}^{u} \tilde{y}^u_{t,k} C_{t,k}(\hat{X}^u) \leq \sum_{k=1}^m \sum_{t=1}^{u} \tilde{y}^u_{t,k} C_{t,k}(\hat{X}^v).
	\end{equation}
\end{lemma}

\begin{proof}
	For $j \in [d]$, let 
	\begin{equation*}
	\tilde{y}^u_{t,k,j} \coloneqq  \begin{cases}
	1 & \text{if $\tilde{y}^u_{t,k} \geq j$} \\
	0 & \text{otherwise}
	\end{cases}
	\end{equation*}
	such that $\tilde{y}^u_{t,k} = \sum_{j=1}^{d} \tilde{y}^u_{t,k,j}$.
	In other words, $\tilde{y}^u_{t,k,j} = 1$ means that the largest server type in the sequence $(\hat{y}^{t'}_{t',k})_{t' \in [t:u]}$ is at least $j$. 
	
	To deduce a contradiction, we assume that there exists a $j \in [d]$ such that 
	\begin{equation} \label{eqn:online:const:opt:scaledcost:assume}
	\sum_{k=1}^m \sum_{t=1}^{u} \tilde{y}^u_{t,k,j} C_{t,k}(\hat{X}^u) > \sum_{k=1}^m \sum_{t=1}^{u}  \tilde{y}^u_{t,k, j} C_{t,k}(\hat{X}^v) .
	\end{equation}
	We consider the schedule $\bar{X}^u$ which is constructed in two steps. First, we insert the schedule $\hat{X}^v$ for all lanes $k$ and time slots $t$ where $\tilde{y}^u_{t,k, j} = 1$ holds into $\hat{X}^u$. Afterwards, we eliminate immediate server changes after $\tilde{y}^u_{t,k, j}$ switches from 1 to 0 by using the greater server type (equivalent to the construction of the \emph{maximum} schedule).  By Lemma~\ref{lemma:online:const:opt:ytilde}, $\tilde{y}^u_{t,k,j} = 1$ implies $\tilde{y}^u_{t',k',j} = 1$ for all $j \in [d]$, $t' \leq t$ and $k' \leq k$, so if $\bar{X}^u$ uses the schedule $\hat{X}^v$ for a given time slot $t$ and lane~$k$, then it also uses $\hat{X}^v$ for the previous time slots $t' \leq t$ and lanes $k' \leq k$. 
	
	The schedule $\bar{X}^u$ is feasible for $\mathcal{I}^u$, because for a given time slot~$t$ and lane~$k$, the server type used by $\hat{X}^v$ is greater than or equal to the server type used by $\hat{X}^u$, so in $\bar{X}^u$ there cannot be more active servers than available. Furthermore the demand requirements are obviously fulfilled. 
	
	The total cost of $\bar{X}^u$ is
	\begin{align*}
	C(\bar{X}^u)  &= \sum_{k=1}^m \sum_{t=1}^{u} \tilde{y}^u_{t,k,j} C_{t,k}(\bar{X}^u) + \sum_{k=1}^m \sum_{t=1}^{u} (1-\tilde{y}^u_{t,k,j}) C_{t,k}(\bar{X}^u) \\
	&\leq \sum_{k=1}^m \sum_{t=1}^{u} \tilde{y}^u_{t,k,j} C_{t,k}(\hat{X}^v) + \sum_{k=1}^m \sum_{t=1}^{u} (1-\tilde{y}^u_{t,k,j}) C_{t,k}(\hat{X}^u) \\
	&\stackrel{\eqref{eqn:online:const:opt:scaledcost:assume}}{<} \sum_{k=1}^m \sum_{t=1}^{u} \tilde{y}^u_{t,k,j} C_{t,k}(\hat{X}^u) + \sum_{k=1}^m \sum_{t=1}^{u} (1-\tilde{y}^u_{t,k,j}) C_{t,k}(\hat{X}^u) \\
	&= C(\hat{X}^u)
	\end{align*}	
	In the first step, we simply split $C(\bar{X}^u)$ into two parts (note that $\tilde{y}^u_{t,k,j}  \in \{0, 1\}$). The first inequality uses the definition of $\bar{X}^u$: for $\tilde{y}^u_{t,k,j} = 1$, we have $C_{t,k}(\bar{X}^u) = C_{t,k}(\hat{X}^v)$ and for $\tilde{y}^u_{t,k,j} = 0$, we have $C_{t,k}(\bar{X}^u) \leq C_{t,k}(\hat{X}^u)$, because $\bar{X}^u$ can use greater server types with lower operating costs than $\hat{X}^u$ due to the elimination of immediate server changes. 
	The last inequality uses our assumption given by equation~\eqref{eqn:online:const:opt:scaledcost:assume}. 
	
	We have shown that $C(\bar{X}^u) < C(\hat{X}^u)$, however, this is a contradiction to the fact that $\hat{X}^u$ is an optimal schedule. Therefore, our assumption was wrong and for all $j \in [d]$,
	\begin{equation*}
	\sum_{k=1}^m \sum_{t=1}^{u} \tilde{y}^u_{t,k,j} C_{t,k}(\hat{X}^u) \leq \sum_{k=1}^m \sum_{t=1}^{u}  \tilde{y}^u_{t,k, j} C_{t,k}(\hat{X}^v) 
	\end{equation*}
	holds. By summarizing these inequalities for all $j \in [d]$ and by using the fact $\sum_{j=1}^{d} \tilde{y}^u_{t,k,j} = \tilde{y}^u_{t,k}$, we get 
	\begin{equation*} 
	\sum_{k=1}^m \sum_{t=1}^{u} \tilde{y}^u_{t,k} C_{t,k}(\hat{X}^u) \leq \sum_{k=1}^m \sum_{t=1}^{u} \tilde{y}^u_{t,k} C_{t,k}(\hat{X}^v). \qedhere
	\end{equation*}
\end{proof}

The following lemma shows how the cost of a single block $A_{v,k}$ can be folded into the term $2\sum_{t=1}^{v-1} \tilde{y}^{v-1}_{t,k} C_{t,k}(\hat{X}^v) $ which is the right hand side of equation~\eqref{eqn:online:const:opt:scaledcost} given in the previous lemma with $u = v-1$.

\begin{lemma} \label{lemma:online:const:opt:mainlemma}
	For all lanes $k \in [m]$ and time slots $v \in [T]$, it is
	\begin{equation} \label{eqn:online:const:opt:mainlemma}
		 2\sum_{t=1}^{v-1} \tilde{y}^{v-1}_{t,k} C_{t,k}(\hat{X}^v) + C(A_{v,k}) \leq  2 \sum_{t=1}^{v} \tilde{y}^v_{t,k} C_{t,k}(\hat{X}^v) .
	\end{equation}
\end{lemma}

\begin{proof}%[Proof of Lemma~\ref{lemma:online:const:opt:mainlemma}]
	If the block $A_{v,k}$ does not exists, equation~\eqref{eqn:online:const:opt:mainlemma} holds by Lemma~\ref{lemma:online:const:opt:ytilde} and $C(A_{v,k}) = 0$. 
	
	If $A_{v,k}$ is a \emph{new} block, then $C(A_{v,k}) \leq 2 \beta_{j}$ with $j \coloneqq j_{v,k} = \hat{y}^v_{v,k}$ by Lemma~\ref{lemma:online:const:alg:costai}. Since $A_{v,k}$ is a \emph{new} block, server type $j$ was not used in the last time slot of the last $\bar{t}_j$ schedules, i.e., $\hat{y}^{t}_{t,k} \leq j - 1$ for $t \in [v - \bar{t}_j : v-1]$. If $\hat{y}^{v-\bar{t}_j}_{v - \bar{t}_j,k} = j$ would hold, then $y^{\mathcal{A}}_{{v-1},k} = j$ and there would be an \emph{extended} block at time slot $v$.
	By using the facts above and the definition of $\tilde{t}^v_{t,k}$, for $t \in [v - \bar{t}_j : v-1]$, we get 
	\begin{align*}
	\tilde{y}^{v-1}_{t,k} &= \max_{t' \in [t:v-1]} \hat{y}^{t'}_{t',k} \leq j-1 = \hat{y}^v_{v,k} - 1 
	\leq  \max_{t' \in [t:v]} \hat{y}^{t'}_{t',k} - 1= \tilde{y}^v_{t,k} - 1 . \numberthis\label{eqn:online:const:opt:mainlemma:ytilde:minusone}
	\end{align*}
	By using Lemma~\ref{lemma:online:const:opt:ytilde} and equation~\eqref{eqn:online:const:opt:mainlemma:ytilde:minusone}, we can estimate the first sum in~\eqref{eqn:online:const:opt:mainlemma}:
	\begin{align*}
	\sum_{t=1}^{v-1} \tilde{y}^{v-1}_{t,k} C_{t,k}(\hat{X}^v) 
	\stackrel{L\ref{lemma:online:const:opt:ytilde},\eqref{eqn:online:const:opt:mainlemma:ytilde:minusone}}{\leq}& 
	\sum_{t=1}^{v - \bar{t}_j - 1} \tilde{y}^{v}_{t,k} C_{t,k}(\hat{X}^v) + \sum_{t=v - \bar{t}_j}^{v-1} (\tilde{y}^{v}_{t,k}-1)C_{t,k}(\hat{X}^v) \\
	%\stackrel{\phantom{\eqref{eqn:online:const:opt:mainlemma:ytilde:vmonotony},\eqref{eqn:online:const:opt:mainlemma:ytilde:minusone}}}{=}&
	%\sum_{t=1}^{v-1} \tilde{y}^{v}_{t,k} C_{t,k}(\hat{X}^v) - \sum_{t=v - \bar{t}_j}^{v-1} C_{t,k}(\hat{X}^v) \\
	%\stackrel{\phantom{L\ref{lemma:online:const:opt:ytilde},\eqref{eqn:online:const:opt:mainlemma:ytilde:minusone}}}{\leq}&
	%\sum_{t=1}^{v} \tilde{y}^{v}_{t,k} C_{t,k}(\hat{X}^v) - \sum_{t=v - \bar{t}_j}^{v} C_{t,k}(\hat{X}^v) \\
	\stackrel{\phantom{L\ref{lemma:online:const:opt:ytilde},\eqref{eqn:online:const:opt:mainlemma:ytilde:minusone}}}{\leq}&
	\sum_{t=1}^{v} \tilde{y}^{v}_{t,k} C_{t,k}(\hat{X}^v) - \beta_j .
	\numberthis\label{eqn:online:const:opt:mainlemma:sumsplit:a}
	\end{align*}
	For the second inequality, we add $(\tilde{y}^v_{v,k} - 1) \cdot C_{v,k}(\hat{X}^v) \geq 0$ and use $\sum_{t=v - \bar{t}_j}^{v}  C_{t,k}(\hat{X}^v)  \geq \beta_j$ which holds because
	either $j$ was powered up in $\hat{X}^v$ during $[v - \bar{t}_j : v]$ (then there is the switching cost of $\beta_j$) or $j$ runs for $\bar{t}_j + 1$ time slots resulting in an operating cost of
	$%\begin{equation*}
	l_j \cdot (\bar{t}_j + 1) = l_j \cdot \left( \left\lfloor \beta_j / l_j \right\rfloor + 1 \right) \geq \beta_j
	$. %\end{equation*}
	Altogether, we get (beginning from the left hand side of equation~\eqref{eqn:online:const:opt:mainlemma} that has to be shown)
	\begin{align*}
	2\sum_{t=1}^{v-1} \tilde{y}^{v-1}_{t,k} C_{t,k}(\hat{X}^v) + C(A_{v,k}) 
	\stackrel{\eqref{eqn:online:const:opt:mainlemma:sumsplit:a},L\ref{lemma:online:const:alg:costai}}{\leq}&
	2\sum_{t=1}^{v} \tilde{y}^{v}_{t,k} C_{t,k}(\hat{X}^v) - 2\beta_j + 2 \beta_j \\
	%= 2\sum_{t=1}^{v} \tilde{y}^{v}_{t,k} C_{t,k}(\hat{X}^v). %\\
	\leq \hspace{8pt} & 2\sum_{t=1}^{v} \tilde{y}^{v}_{t,k} C_{t,k}(\hat{X}^v).
	\end{align*}
	
	If $A_{v,k}$ is an \emph{extended} block, then $C(A_{v,k}) \leq l_j d$ with $j \coloneqq j_{v,k}$ and $d \coloneqq d_{t,k}$ by Lemma~\ref{lemma:online:const:alg:costai}. 
	Let $j' \coloneqq \hat{y}^v_{v,k}$ be the server type in $\hat{X}^v$ that provoked the extended block. For each $t \in [v - d + 1 : v - 1]$, $\hat{y}^t_{t,k} \leq j' - 1$ holds, because otherwise the duration of $A_{v,k}$ would be smaller than~$d$. 
	%The duration $d$ of the block is exactly one greater than the number of consecutive time slots where $\hat{y}^t_{t,k} < j$ for $t = v - 1$ down to $0$, i.e., $\hat{y}^t_{t,k} \leq j-1$ for $t \in [v - d + 1 : v - 1]$.
	Analogously to \emph{new} blocks, equation~\eqref{eqn:online:const:opt:mainlemma:ytilde:minusone} holds for all $t \in [v - d + 1 : v-1]$.
	The first sum of equation~\eqref{eqn:online:const:opt:mainlemma} is at most
	\begin{align*}
	\sum_{t=1}^{v-1} \tilde{y}^{v-1}_{t,k} C_{t,k}(\hat{X}^v) 
	\stackrel{L\ref{lemma:online:const:opt:ytilde},\eqref{eqn:online:const:opt:mainlemma:ytilde:minusone}}{\leq}& \sum_{t=1}^{v - d} \tilde{y}^{v}_{t,k} C_{t,k}(\hat{X}^v) + \sum_{t=v - d + 1}^{v-1} (\tilde{y}^{v}_{t,k}-1)C_{t,k}(\hat{X}^v) \\
	\stackrel{\phantom{(11),(11)}}{=}& \sum_{t=1}^{v-1} \tilde{y}^{v}_{t,k} C_{t,k}(\hat{X}^v) - \sum_{t=v - d + 1}^{v-1} C_{t,k}(\hat{X}^v)  \\
	\stackrel{\phantom{(11),(11)}}{\leq}& \sum_{t=1}^{v} \tilde{y}^{v}_{t,k} C_{t,k}(\hat{X}^v) - \sum_{t=v - d + 1}^{v} C_{t,k}(\hat{X}^v). \numberthis\label{eqn:online:const:opt:mainlemma:sumsplit:b}
	\end{align*}
	
	The last term in \eqref{eqn:online:const:opt:mainlemma:sumsplit:b} satisfies
	\begin{equation}\label{eqn:online:const:opt:mainlemma:sumcost:b}
	\sum_{t=v - d + 1}^{v} C_{t,k}(\hat{X}^v) \geq l_{j'} d,
	\end{equation}
	because either $j'$ runs for $d$ time slots in $\hat{X}^v$ during $[v - d + 1: v]$ (then the operating cost is exactly $l_{j'} d$) or $j'$ was powered up during this interval resulting in a cost of
	\begin{equation*}
	\beta_{j'} \geq l_{j'} \left\lfloor \frac{\beta_{j'}}{l_{j'}} \right\rfloor = l_{j'} \bar{t}_{j'} \geq l_{j'} d 
	\end{equation*}
	as the duration $d$ of block $A_{v,k}$ is upper bounded by $\bar{t}_j$. 
	
	Altogether, we get
	\begin{align*}
	& 2\sum_{t=1}^{v-1} \tilde{y}^{v-1}_{t,k} C_{t,k}(\hat{X}^v) + C(A_{v,k}) \\
	\stackrel{\eqref{eqn:online:const:opt:mainlemma:sumsplit:b},L\ref{lemma:online:const:alg:costai}}{\leq}& 2\sum_{t=1}^{v} \tilde{y}^{v}_{t,k} C_{t,k}(\hat{X}^v) - 2\sum_{t=v - d + 1}^{v} C_{t,k}(\hat{X}^v) +  l_j d \\
	\stackrel{\eqref{eqn:online:const:opt:mainlemma:sumcost:b}}{\leq} \hspace{8pt} & 2\sum_{t=1}^{v} \tilde{y}^{v}_{t,k} C_{t,k}(\hat{X}^v). 
	\end{align*}
	The last inequality holds, because $j' = \hat{y}^v_{v,k} \leq y^\mathcal{A}_{v,k} = j$ implies $l_{j'} d \geq l_j d$. \qedllncs
\end{proof}

%%THEOREM: 2d-competitive
\begin{theorem} \label{theorem:online:const}
	Algorithm~$\mathcal{A}$ is $2d$-competitive.
\end{theorem}

\begin{proof}%[Proof of Theorem~\ref{theorem:online:const}]
	The feasibility of $X^\mathcal{A}$ was already proven in Lemma~\ref{lemma:online:const:alg:feasible}, so we have to show that $C(X^\mathcal{A}) \leq 2d \cdot C(\hat{X}^T)$.
	Let 
	%\begin{equation*}
	$C_v(X^\mathcal{A}) \coloneqq \sum_{t=1}^{v} \sum_{k=1}^m C(A_{t,k})$
	%\end{equation*}
	denote the cost of algorithm~$\mathcal{A}$ up to time slot $v$. We will show by induction that 
	\begin{equation} \label{eqn:online:const:induction}
	C_v(X^\mathcal{A}) \leq 2\sum_{k=1}^m \sum_{t=1}^{v} \tilde{y}^v_{t,k} C_{t,k}(\hat{X}^v)
	\end{equation} 
	holds for all $v \in [T]_0$. 
	
	For $v = 0$, we have no costs for both $X^\mathcal{A}$ and $\hat{X}^v$, so inequality~\eqref{eqn:online:const:induction} is fulfilled. Assume that inequality~\eqref{eqn:online:const:induction} holds for $v-1$. By using the induction hypothesis as well as Lemmas~\ref{lemma:online:const:opt:scaledcost} and~\ref{lemma:online:const:opt:mainlemma}, we get
	%{
	%\setlength{\belowdisplayskip}{0pt}
	\begin{align*}
	C_v(X^\mathcal{A}) &\stackrel{\phantom{L7,L8}}{=}
	C_{v-1}(X^\mathcal{A})  + \sum_{k=1}^m C(A_{v,k}) \\
	&\hspace{3pt}\stackrel{\text{I.H.}}{\leq}\hspace{3pt}
	2\sum_{k=1}^m \sum_{t=1}^{v-1} \tilde{y}^{v-1}_{t,k} C_{t,k}(\hat{X}^{v-1}) +  \sum_{k=1}^m C(A_{v,k}) \\
	%&\stackrel{L\ref{lemma:online:const:opt:scaledcost}}{\leq} 2\sum_{k=1}^m \sum_{t=1}^{v-1} \tilde{y}^{v-1}_{t,k} C_{t,k}(\hat{X}^v) +  \sum_{k=1}^m C(A_{v,k}) \\
	&\stackrel{L\ref{lemma:online:const:opt:scaledcost},L\ref{lemma:online:const:opt:mainlemma}}{\leq } 2\sum_{k=1}^m \sum_{t=1}^{v} \tilde{y}^v_{t,k} C_{t,k}(\hat{X}^v) . \numberthis\label{eqn:online:const:induction:cvxa}
	\end{align*}
	%}
	Since $\tilde{y}^v_{t,k}  \leq d$, we get
	\begin{align*}
	%C(X^\mathcal{A}) %&\stackrel{\phantom{(11)}}{=} C_T(X^\mathcal{A})  \\
	%= 
	C_T(X^\mathcal{A})
	\stackrel{\eqref{eqn:online:const:induction:cvxa}}{\leq} 2\sum_{k=1}^m \sum_{t=1}^{T} \tilde{y}^T_{t,k} C_{t,k}(\hat{X}^T) %\\
	\stackrel{\phantom{(20)}}{\leq} 
	2d \sum_{k=1}^m \sum_{t=1}^{T} C_{t,k}(\hat{X}^T) %\\
	&\stackrel{\phantom{(20)}}{\leq} 
	2d \cdot C(\hat{X}^T).
	\end{align*}
	The schedule $\hat{X}^T$ is optimal for the problem instance $\mathcal{I}$, so algorithm~$\mathcal{A}$ is $2d$-competitive. \qedllncs
\end{proof}

\section{Randomized Online Algorithm}
\label{sec:online:rand}

The $2d$-competitive algorithm can be randomized to achieve a competitive ratio of $\frac{e}{e-1} d \approx 1.582 d$ against an oblivious adversary.
The randomized algorithm $\mathcal{B}$ chooses $\gamma \in [0,1]$ according to the probability density function $f_\gamma(x) = e^x / (e-1)$ for $x \in [0,1]$. 
The variables $\bar{t}_j$ are set to $\left\lfloor \gamma \cdot \beta_j / l_j \right\rfloor$, so the running time of a server is randomized. Then, algorithm $\mathcal{A}$ is executed. Note that $\gamma$ is determined at the beginning of the algorithm and not for each block.

Lemmas~\ref{lemma:online:const:opt:laneswitching}-\ref{lemma:online:const:alg:feasible} as well as~\ref{lemma:online:const:opt:ytilde} and~\ref{lemma:online:const:opt:scaledcost} still hold, because they do not depend on the exact value of~$\bar{t}_j$.
Only Lemmas~\ref{lemma:online:const:alg:costai}  and~\ref{lemma:online:const:opt:mainlemma} have to be adapted.
First of all, we have to introduce a new variable. Let
\begin{equation*}
\hat{\tau}_{t,k} \coloneqq \max \left\{\tau \in [\bar{t}_{\hat{y}^t_{t,k}}] \bigm\vert \forall \tau' \in [\tau-1] : \hat{y}^{t-\tau'}_{t-\tau',k} < \hat{y}^t_{t,k}\right\}
\end{equation*} be the number of time slots we have to go backwards in time to find an optimal schedule $\hat{X}^{t-\tau}$ that uses a server type greater than or equal to $\hat{y}^t_{t,k}$ in its last time slot in lane $k$.
The following lemma replaces Lemma~\ref{lemma:online:const:alg:costai} and estimates the expected cost of the block $A_{t,k}$ depending on $\hat{\tau}_{t,k}$.

\begin{lemma} \label{lemma:online:const:rand:costai}
	Let $c = e/(e-1)$, $j \coloneqq \hat{y}^t_{t,k}$ and $\tau \coloneqq \hat{\tau}_{t,k}$
	The expected cost of the block $A_{t,k}$ is upper bounded by
	$%\begin{equation*}
	\mathbb{E}[C(A_{t,k})] \leq l_j \tau c
	$. %\end{equation*}
\end{lemma}

\begin{proof}%[Proof of Lemma~\ref{lemma:online:const:rand:costai}]
	\newcommand{\gammathreshold}{q}
	Let $\gammathreshold \coloneqq \frac{l_j}{\beta_j} \tau$ (note that both $j$ and $\tau$ do not depend on random decisions). We estimate the cost of $A_{t,k}$ depending on $\gamma$.
	
	If $\gamma > \gammathreshold$, then the server $y^\mathcal{B}_{t-\tau,k} \geq j$ is still running at time slot~$t$, since $\hat{y}^{t-\tau}_{t-\tau,k} \geq \hat{y}^t_{t,k} = j$ implies 
	\begin{equation*}
	\bar{t}_{\hat{y}^{t-\tau}_{t-\tau,k}} \geq \bar{t}_j = \left\lfloor \gamma \cdot \beta_j / l_j  \right\rfloor > \left\lfloor \gammathreshold \cdot  \beta_j / l_j \right\rfloor \geq \tau.
	\end{equation*}
	Therefore, $A_{t,k}$ is an \emph{extended} block with duration at most $\tau$ (or $A_{t,k}$ does not exists which is equivalent to an extended block with duration~0). Furthermore, server type $j_{t,k} = y^\mathcal{B}_{t-\tau,k}$ used in $A_{t,k}$ is greater than or equal to $j = \hat{y}^t_{t,k}$, so $l_{j_{t,k}} \leq l_j$. Thus, for $\gamma > \gammathreshold$, we have $C(A_{t,k}) \leq l_j \tau$.
	
	If $\gamma \leq \gammathreshold$, then there can be a \emph{new} block at time slot $t$. Note that this is only a necessary, not a sufficient condition for a \emph{new} block (e.g., if $\hat{y}^{t-\tau-1}_{t-\tau-1,k} > \hat{y}^t_{t,k}$). If $A_{t,k}$ is a \emph{new} block, then its cost is given by $\beta_j + l_j \bar{t}_j$. If $y^\mathcal{B}_{t-\tau,k}$ still runs at time slot $t$, then $A_{t,k}$ is an \emph{extended} block whose cost is at most $l_j \bar{t}_j$, since $j \leq j_{t,k}$. Thus, for $\gamma \leq \gammathreshold$, we have $C(A_{t,k}) \leq \beta_j + l_j \bar{t}_j = \beta_j + l_j \left\lfloor \gamma \cdot \beta_j / l_j  \right\rfloor$. 
	
	Now, we can estimate the expected cost of $A_{t,k}$ by using the density function~$f_\gamma$.
	\begin{align*}
	\mathbb{E}[C(A_{t,k})] &\leq 
	\int_{0}^{\gammathreshold} f_\gamma (x) \left( \beta_j + l_j \left\lfloor x \cdot \frac{\beta_j}{l_j} \right\rfloor \right) \diff x
	+ \int_{\gammathreshold}^{1} f_\gamma (x) l_j \tau \diff x \\
	&\leq 
	%\beta_j \int_{0}^{\gammathreshold} f_\gamma (x) \left(1 + x  \right) \diff x
	%+ l_j \tau \int_{\gammathreshold}^{1} f_\gamma (x) \diff x \\
	%&=
	\beta_j \left( \gammathreshold \cdot F_\gamma (\gammathreshold) + \frac{\gammathreshold}{e-1} \right) + l_j \tau \big(1 - F_\gamma (\gammathreshold)\big) .
	%&= F_\gamma(\gammathreshold) \left( \beta_j \gammathreshold - l_j \tau \right) + \beta_j \frac{\gammathreshold}{e-1} + l_j \tau
	\end{align*}%
	The last inequality uses $l_j \left\lfloor x \cdot \beta_j / l_j  \right\rfloor \leq \beta_j x$, so the integrals can easily be calculated. By using $\beta_j q = l_j \tau$ (which follows from the definition of $\gammathreshold$), we get 
	\begin{align*}
	\mathbb{E}[C(A_{t,k})] &\leq \beta_j \left( F_\gamma (\gammathreshold) \gammathreshold + \frac{\gammathreshold}{e-1} \right) + l_j \tau (1 - F_\gamma (\gammathreshold)) \\
	&= l_j \tau \left(\frac{1}{e-1} + 1\right) \\
	&= l_j \tau c. \qedhere
	\end{align*}
	
\end{proof}

The following lemma replaces Lemma~\ref{lemma:online:const:opt:mainlemma} and shows how the expected cost of block $A_{v,k}$ can be folded into the term $c \cdot \sum_{t=1}^{v-1} \tilde{y}^{v-1}_{t,k} C_{t,k}(\hat{X}^v)$ which is the right hand side of equation~\eqref{eqn:online:const:opt:scaledcost}. 

\begin{lemma} \label{lemma:online:const:rand:mainlemma}
	For all lanes $k \in [m]$ and time slots $v \in [T]$, it holds
	\begin{equation} \label{eqn:online:const:rand:mainlemma}
	c \cdot \sum_{t=1}^{v-1} \tilde{y}^{v-1}_{t,k} C_{t,k}(\hat{X}^v) + \mathbb{E}[C(A_{v,k})] \leq  c \cdot \sum_{t=1}^{v} \tilde{y}^v_{t,k} C_{t,k}(\hat{X}^v) .
	\end{equation}
\end{lemma}

\begin{proof}%[Proof of Lemma~\ref{lemma:online:const:rand:mainlemma}]
	If $\hat{y}^v_{v,k} = 0$, then $A_{v,k}$ does not exist, so $\mathbb{E}[C(A_{v,k})] = 0$ and therefore equation~\eqref{eqn:online:const:rand:mainlemma} holds by $\tilde{y}^{v-1}_{t,k} \leq \tilde{y}^v_{t,k}$ (see Lemma~\ref{lemma:online:const:opt:ytilde}).
	
	Thus, in the following we consider the case $\hat{y}^v_{v,k} > 0$. 
	For all $t \in [v - \tau + 1 : v - 1]$ with $\tau \coloneqq \hat{\tau}_{v,k}$, the inequality $\tilde{y}^{v-1}_{t,k} \leq \tilde{y}^v_{t,k} - 1$ holds (see equation~\eqref{eqn:online:const:opt:mainlemma:ytilde:minusone} in the proof of Lemma~\ref{lemma:online:const:opt:mainlemma}). Therefore, we get
	\begin{align*}
	& \sum_{t=1}^{v-1} \tilde{y}^{v-1}_{t,k} C_{t,k}(\hat{X}^v) \\
	\stackrel{L\ref{lemma:online:const:opt:ytilde},\eqref{eqn:online:const:opt:mainlemma:ytilde:minusone}}{\leq} &
	\sum_{t=1}^{v-\tau} \tilde{y}^{v}_{t,k} C_{t,k}(\hat{X}^v) + \sum_{t=v-\tau+1}^{v-1} (\tilde{y}^{v}_{t,k} - 1) C_{t,k}(\hat{X}^v) \\
	\stackrel{\phantom{L\ref{lemma:online:const:opt:ytilde},\eqref{eqn:online:const:opt:mainlemma:ytilde:minusone}}}{\leq}&
	\sum_{t=1}^{v} \tilde{y}^{v}_{t,k} C_{t,k}(\hat{X}^v) - \sum_{t = v-\tau+1}^{v} C_{t,k}(\hat{X}^v) \numberthis\label{eqn:online:const:rand:mainlemma:a}
	\end{align*}
	For the last inequality, we add the term $(\tilde{y}^v_{v,k} - 1) C_{t,k}(\hat{X}^v)$ which is positive, since $\hat{y}^v_{v,k} > 0$. 
	
	The last term in \eqref{eqn:online:const:rand:mainlemma:a} satisfies
	\begin{equation} \label{eqn:online:const:rand:mainlemma:b}
	\sum_{t = v-\tau+1}^{v} C_{t,k}(\hat{X}^v) \geq l_j \tau
	\end{equation}
	with $j \coloneqq \hat{y}^v_{v,k}$, because either $j$ runs for $\tau$ time slots in $\hat{X}^v$ or $j$ is powered up during $[v - \tau + 1 : v]$ resulting in a cost of 
	\begin{equation*}
	\beta_j \geq l_j \left\lfloor \frac{\beta_j}{l_j} \gamma \right\rfloor = l_j \bar{t}_j \geq l_j \tau ,
	\end{equation*}
	as $\tau \leq \bar{t}_j$ by definition.
	
	By using Lemma~\ref{lemma:online:const:rand:costai}, we get
	\begin{align*}
	& c \cdot \sum_{t=1}^{v-1} \tilde{y}^{v-1}_{t,k} C_{t,k}(\hat{X}^v) + \mathbb{E}[C(A_{v,k})] \\
	\stackrel{\eqref{eqn:online:const:rand:mainlemma:a},L\ref{lemma:online:const:rand:costai}}{\leq} & c \cdot \sum_{t=1}^{v} \tilde{y}^{v}_{t,k} C_{t,k}(\hat{X}^v) - c \cdot \sum_{t = v-\tau+1}^{v} C_{t,k}(\hat{X}^v) + l_j \tau \cdot c \\
	\stackrel{\eqref{eqn:online:const:rand:mainlemma:b}}{\leq} \hspace{8pt} & c \cdot \sum_{t=1}^{v} \tilde{y}^T_{t,k} C_{t,k}(\hat{X}^v) \qedhere
	\end{align*}
\end{proof}

\begin{theorem} \label{theo:online:const:rand}
	Algorithm $\mathcal{B}$ is $\frac{e}{e-1} d$-competitive against an oblivious adversary.
\end{theorem}

\begin{proof}%[Proof of Theorem~\ref{theo:online:const:rand}]
	Lemma~\ref{lemma:online:const:alg:feasible} still holds for algorithm $\mathcal{B}$, so the schedule~$X^\mathcal{B}$ is feasible.
	We have to show that $\mathbb{E}[C(X^\mathcal{B})] \leq c d \cdot C(\hat{X}^T)$ with $c = \frac{e}{e-1}$. Let 
	\begin{equation*}
	E_v(X^\mathcal{B}) \coloneqq \mathbb{E}\left[ \sum_{t=1}^{v} \sum_{k=1}^{m} C(A_{v,k}) \right] = \sum_{t=1}^{v} \sum_{k=1}^{m} \mathbb{E}[C(A_{v,k})]
	\end{equation*}
	denote the expected cost of algorithm~$\mathcal{B}$ up to time slot $v$. We will show by induction that
	\begin{equation} \label{eqn:online:const:rand:induction}
	E_v(X^\mathcal{B}) \leq c \sum_{k=1}^{m} \sum_{t=1}^{v} \tilde{y}^v_{t,k} C_{t,k}(\hat{X}^v)
	\end{equation}
	holds for all $v \in [T]_0$. 
	
	For $v = 0$, we have no costs for both $X^\mathcal{B}$ and $\hat{X}^v$, so inequality~\eqref{eqn:online:const:rand:induction} is fulfilled. Assume that inequality~\eqref{eqn:online:const:rand:induction} holds for $v-1$. By using the induction hypothesis as well as Lemmas~\ref{lemma:online:const:opt:scaledcost} and~\ref{lemma:online:const:rand:mainlemma}, we get
	\begin{align*}
	E_v(X^\mathcal{B}) &\stackrel{\phantom{L2.13}}{=}E_{v-1}(X^\mathcal{B})  + \sum_{k=1}^m \mathbb{E}[C(A_{v,k})] \\
	&\stackrel{\hspace{3pt}\text{I.H.}\hspace{3pt}}{\leq} c \sum_{k=1}^m \sum_{t=1}^{v-1} \tilde{y}^{v-1}_{t,k} C_{t,k}(\hat{X}^{v-1}) +  \sum_{k=1}^m \mathbb{E}[C(A_{v,k})] \\
	&\stackrel{L\ref{lemma:online:const:opt:scaledcost}}{\leq} c\sum_{k=1}^m \sum_{t=1}^{v-1} \tilde{y}^{v-1}_{t,k} C_{t,k}(\hat{X}^v) +  \sum_{k=1}^m \mathbb{E}[C(A_{v,k})] \\
	&\stackrel{L\ref{lemma:online:const:rand:mainlemma}}{\leq } c\sum_{k=1}^m \sum_{t=1}^{v} \tilde{y}^v_{t,k} C_{t,k}(\hat{X}^v) . \numberthis\label{eqn:online:const:rand:induction:ev}
	\end{align*}
	
	Since $\tilde{y}^v_{t,k} \leq d$, we get
	\begin{align*}
	\mathbb{E}[C(X^\mathcal{B})]
	= E_T(X^\mathcal{B})
	&\stackrel{\eqref{eqn:online:const:rand:induction:ev}}{\leq} c\sum_{k=1}^m \sum_{t=1}^{T} \tilde{y}^T_{t,k} C_{t,k}(\hat{X}^T) \\
	&\stackrel{\phantom{(11)}}{\leq} cd \sum_{k=1}^m \sum_{t=1}^{T} C_{t,k}(\hat{X}^T) \\
	&\stackrel{\phantom{(11)}}{\leq} cd \cdot C(\hat{X}^T).
	\end{align*}
	The schedule $\hat{X}^T$ is optimal for the problem instance $\mathcal{I}$, so algorithm~$\mathcal{B}$ is $cd$-competitive. \qedllncs
\end{proof}

%The complete proof of this theorem is shown in the full paper.
%Most lemmas introduced in the previous section still hold, because they do not depend on the exact value of $\bar{t}_j$, only Lemmas~\ref{lemma:online:const:alg:costai}  and~\ref{lemma:online:const:opt:mainlemma} have to be adapted.
%For the proof of Theorem~\ref{theo:online:const:rand}, we first give an upper bound for the expected cost of block $A_{t,k}$ (replacing Lemma~\ref{lemma:online:const:alg:costai}). This bound is used to show that 
%\begin{equation*}
%\frac{e}{e-1} \cdot \sum_{t=1}^{v-1} \tilde{y}^{v-1}_{t,k} C_{t,k}(\hat{X}^v) + \mathbb{E}[C(A_{v,k})] \leq  \frac{e}{e-1} \cdot \sum_{t=1}^{v} \tilde{y}^v_{t,k} C_{t,k}(\hat{X}^v)
%\end{equation*} holds for all lanes $k \in [m]$ and time slots $v \in [T]$ (similar to Lemma~\ref{lemma:online:const:opt:mainlemma}). Finally, Theorem~\ref{theo:online:const:rand} is proven by induction.

\section{Lower bound}
\label{sec:online:lower}

In this section, we show that there is no deterministic online algorithm that achieves a competitive ratio that is better than $2d$.

We consider the following problem instance: Let $\beta_j \coloneqq N^{2j}$ and $l_j \coloneqq 1 /N^{2j}$ where $N$ is a sufficiently large number that depends on the number of servers types $d$. The value of $N$ will be determined later. The adversary will send a job for the current time slot if and only if the online algorithm has no active server during the previous time slot. This implies that the online algorithm has to power up a server immediately after powering down any server. Note that $\lambda_t \in \{0,1\}$, i.e., it is never necessary to power up more than one server. The optimal schedule is denoted by $X^\ast$. Let $\mathcal{A}$ be an arbitrary deterministic online algorithm and let $X^\mathcal{A}$ be the schedule computed by $\mathcal{A}$. 

W.l.o.g., in $X^\mathcal{A}$ there is no time slot with more than one active server. If this were not the case, we could easily convert the schedule into one where the assumption holds without increasing the cost. Assume that at time slot $t$ a new server of type $k$ is powered up such that there are (at least) two active servers at time $t$. If we power up the server at $t+1$, the schedule is still feasible, but the total costs are reduced by $l_k$. We can repeat this procedure until there is at most one active server for each time slot.

\begin{lemma} \label{lemma:online:lower:induction}
	Let $k \in [d]$. If $X^\mathcal{A}$ only uses servers of type lower than or equal to $k$ and if the cost of $\mathcal{A}$ is at least $C(X^\mathcal{A}) \geq {N} \beta_k$, then the cost of $\mathcal{A}$ is at least
	\begin{equation} \label{eqn:online:lower:induction}
	C(X^\mathcal{A}) \geq (2k - \epsilon_k) \cdot C(X^\ast)
	\end{equation}
	with $\epsilon_k = 9k^2 / {N}$ and $N \geq 6k$.
\end{lemma}

\begin{proof}%[Proof of Lemma~\ref{lemma:online:lower:induction}]
	We will prove the lemma by induction. 
	
	For $k = 1$, let $t$ be the length of the schedule $X^\mathcal{A}$ and let $n$ denote how often server type 1 is powered up in $X^\mathcal{A}$. The cost of $X^\mathcal{A}$ is $C(X^\mathcal{A}) = n \beta_1 + l_1 (t - n + 1)$. We use two strategies to estimate the cost of an optimal schedule. In the first strategy the server runs for the whole time, so the cost is $\beta_1 + l_1 t$. The second strategy is to power down the server when it is idle, so the cost is $n(\beta_1 + l_1)$.
	
	We differ between the cases $n \geq N/8$ (case 1) and $n < N/8$ (case~2).
	In case~1, the competitive ratio is
	\begin{align*}
	\frac{C(X^\mathcal{A})}{C(X^\ast)} &= \frac{n \beta_1 + l_1 (t-n+1)}{C(X^\ast)} \\
	&\geq \frac{(n-1)\beta_1 - l_1 (n-1)}{n(\beta_1 + l_1)} + \frac{\beta_1 + l_1 t}{\beta_1 + l_1 t} \\
	&= \left(1 - \frac{1}{n} \right) - \frac{2l_1 (n-1)}{n(\beta_1 + l_1)} + 1 . \numberthis\label{lemma:online:lower:induction:one:a}
	\end{align*}
	For the inequality, we split the cost of $X^\mathcal{A}$ into two terms and estimate the cost of $X^\ast$ in left quotient with the second strategy and $C(X^\ast)$ in the right quotient with the first strategy.
	
	The quotient $\frac{2l_1 (n-1)}{n(\beta_1 + l_1)}$ can be estimated by using $n(\beta_1 + l_1) \geq (n-1) \beta_1$, the definitions of $l_1$ and $\beta_1$ as well as the precondition of the lemma that requires $N \geq 6k$.
	\begin{equation*}
	\frac{2l_1 (n-1)}{n(\beta_1 + l_1)} \leq \frac{2l_1 (n-1)}{(n-1) \beta_1} \leq \frac{2}{N^4} < \frac{1}{N}.
	\end{equation*}
	
	By using this result in equation~\eqref{lemma:online:lower:induction:one:a} as well as $n \geq N/8$, we get
	\begin{align*}
	\frac{C(X^\mathcal{A})}{C(X^\ast)} > 2 - \frac{8}{N} - \frac{1}{N} = 2 - \epsilon_1 ,
	\end{align*}
	since $\epsilon_1 = 9/N$.
	
	In case 2, we use the fact that $C(X^\mathcal{A}) \geq N \beta_1$, so the competitive ratio is at least
	\begin{align*}
	\frac{C(X^\mathcal{A})}{C(X^\ast)} \geq \frac{N\beta_1}{n(\beta_1 + l_1)} 
	> \frac{8n \beta_1}{2n\beta_1} 
	&= 4 .
	\end{align*}
	In the first inequality, we use the second strategy to estimate the cost of $C(X^\ast)$. The second inequality holds because $n < N/8$ and $n(\beta_1 + l_1) \leq 2n\beta_1$. 
	
	For both cases, we have shown that $C(X^\mathcal{A}) \geq (2 - \epsilon_1) C(X^\ast)$ holds, so equation~\eqref{eqn:online:lower:induction} is fulfilled for $k=1$.

	Next, assume that Lemma~\ref{lemma:online:lower:induction} holds for $k-1$.
	
	We divide the schedule $X^\mathcal{A}$ into phases $L_0, K_1, L_1, K_2, \dots, L_n$ such that in the phases $K_1, \dots, K_n$ server type $k$ is used exactly once, while in the intermediate phases $L_0, \dots, L_n$ the other server types $1, \dots, k-1$ are used. A phase~$K_i$ begins when a server of type $k$ is powered up and ends when it is powered down. The phases $L_i$ can have zero length (if the server type $k$ is powered up immediately after it is powered down, so between $K_i$ and $K_{i+1}$ an empty phase~$L_i$ is inserted). 
	
	The operating cost during phase $K_i$ is denoted by $\delta_i \beta_k$. 
	The operating and switching costs during phase $L_i$ are denoted by $p_i \beta_k$.
	We divide the intermediate phases $L_i$ into long phases where $p_i > 1 / {N}$ holds and short phases where $p_i \leq 1 / {N}$. Note that we can use the induction hypothesis only for long phases. The index sets of the long and short phases are denoted by $\mathcal{L}$ and $\mathcal{S}$, respectively.
	
		\ifincludefigures
		\begin{figure}[t]
	\setlength{\abovecaptionskip}{9pt plus 0pt minus 0pt}
	\setlength{\belowcaptionskip}{0pt plus 0pt minus 0pt}
	
	\centering
	\begin{tikzpicture}
		\pgfmathsetmacro{\xBegin}{0}
		\pgfmathsetmacro{\xStep}{0.67}
		
		\pgfmathsetmacro{\yBeginAlg}{0}
		\pgfmathsetmacro{\yBeginLambda}{0.9}
		\pgfmathsetmacro{\yBeginOptA}{1.8}
		\pgfmathsetmacro{\yBeginOptB}{-0.9}
		\pgfmathsetmacro{\ySize}{0.6}
		\pgfmathsetmacro{\yLambda}{0.6}
		
		\tikzstyle{llong}=[fill=blue!20!white]
		\tikzstyle{lshort}=[fill=green!20!white]
		\tikzstyle{optone}=[fill=red!20!white]
		
		\draw[llong] (\xBegin + 0*\xStep, \yBeginAlg) rectangle (\xBegin + 2*\xStep, \yBeginAlg + \ySize)
			node[pos=.5] {$L_0$};
		\draw (\xBegin + 2*\xStep, \yBeginAlg) rectangle (\xBegin + 6*\xStep, \yBeginAlg + \ySize)
			node[pos=.5] {$K_1$};
		\draw (\xBegin + 6*\xStep, \yBeginAlg) rectangle (\xBegin + 8.5*\xStep, \yBeginAlg + \ySize)
			node[pos=.5] {$K_2$};
		\draw[lshort] (\xBegin + 8.5*\xStep, \yBeginAlg) rectangle (\xBegin + 9.3*\xStep, \yBeginAlg + \ySize)
			node[pos=.5] {$L_2$};
		\draw (\xBegin + 9.3*\xStep, \yBeginAlg) rectangle (\xBegin + 13.5*\xStep, \yBeginAlg + \ySize)
			node[pos=.5] {$K_3$};
		\draw[llong] (\xBegin + 13.5*\xStep, \yBeginAlg) rectangle (\xBegin + 15.2*\xStep, \yBeginAlg + \ySize)
			node[pos=.5] {$L_3$};
			
		\draw (\xBegin, \yBeginLambda) 
			to (\xBegin + 2*\xStep, \yBeginLambda)
			to (\xBegin + 2*\xStep, \yBeginLambda + \yLambda)
			to (\xBegin + 2.2*\xStep, \yBeginLambda + \yLambda)
			to (\xBegin + 2.2*\xStep, \yBeginLambda)
			to (\xBegin + 6*\xStep, \yBeginLambda)
			to (\xBegin + 6*\xStep, \yBeginLambda + \yLambda)
			to (\xBegin + 6.2*\xStep, \yBeginLambda + \yLambda)
			to (\xBegin + 6.2*\xStep, \yBeginLambda)
			to (\xBegin + 9.3*\xStep, \yBeginLambda)
			to (\xBegin + 9.3*\xStep, \yBeginLambda + \yLambda)
			to (\xBegin + 9.5*\xStep, \yBeginLambda + \yLambda)
			to (\xBegin + 9.5*\xStep, \yBeginLambda)
			to (\xBegin + 15.2*\xStep, \yBeginLambda);
			
		\foreach \x in {0, 0.1,...,1.81} {
			\draw (\xBegin + \x * \xStep, \yBeginLambda)
				to (\xBegin + \x * \xStep, \yBeginLambda + \yLambda);
		}
		\foreach \x in {8.5, 8.6,...,9.11} {
			\draw (\xBegin + \x * \xStep, \yBeginLambda)
			to (\xBegin + \x * \xStep, \yBeginLambda + \yLambda);
		}
		\foreach \x in {13.5, 13.6,...,15.01} {
			\draw (\xBegin + \x * \xStep, \yBeginLambda)
			to (\xBegin + \x * \xStep, \yBeginLambda + \yLambda);
		}
			
		\draw (\xBegin + 0 * \xStep, \yBeginOptA) rectangle (\xBegin + 15.2*\xStep, \yBeginOptA + \ySize)
			node[pos=.5] {$y_{t,k} = k$};
			
		\draw[] (\xBegin + 0 * \xStep, \yBeginOptB) rectangle (\xBegin + 2*\xStep, \yBeginOptB + \ySize)
			node[pos=.5] {$1/\alpha$};
		\draw[optone] (\xBegin + 2 * \xStep, \yBeginOptB) rectangle (\xBegin + 2.2*\xStep, \yBeginOptB + \ySize);
		\draw[optone] (\xBegin + 6 * \xStep, \yBeginOptB) rectangle (\xBegin + 6.2*\xStep, \yBeginOptB + \ySize);
		\draw[] (\xBegin + 8.5 * \xStep, \yBeginOptB) rectangle (\xBegin + 9.3*\xStep, \yBeginOptB + \ySize)
			node[pos=.5] {$\mathcal{A}$};
		\draw[optone] (\xBegin + 9.3 * \xStep, \yBeginOptB) rectangle (\xBegin + 9.5*\xStep, \yBeginOptB + \ySize);
		\draw[] (\xBegin + 13.5 * \xStep, \yBeginOptB) rectangle (\xBegin + 15.2*\xStep, \yBeginOptB + \ySize)
			node[pos=.5] {$1/\alpha$};
			
		\node[left] at (\xBegin, \yBeginAlg + 0.5 * \ySize) {$X^{\mathcal{A}}$};
		\node[left] at (\xBegin, \yBeginLambda + 0.5 * \yLambda) {$\lambda_t$};
		\node[left] at (\xBegin, \yBeginOptA + 0.5 * \ySize) {Strategy 1};
		\node[left] at (\xBegin, \yBeginOptB + 0.5 * \ySize) {Strategy 2};
		
	\end{tikzpicture}
	\caption{{\normalfont (figure is colored)} Visualization of the two strategies to estimate the cost of an optimal schedule. \normalfont The schedule of algorithm~$\mathcal{A}$ and the incoming jobs $\lambda_t$ are shown in the middle. Long phases are marked in blue and short phases are marked in green ($L_1$ is a short phase with zero length). Strategy 1 simply uses server type $k$ the whole time. During the short phases, strategy 2 behaves like algorithm~$\mathcal{A}$. For the long phases, there is a solution that results in only $1/\alpha$ of the cost of $X^{\mathcal{A}}$ with $\alpha \coloneqq 2k - 2 - \epsilon_{k-1}$. In the red blocks server type 1 is activated for exactly one time slot. }
	\label{fig:online:lower:induction}
\end{figure}
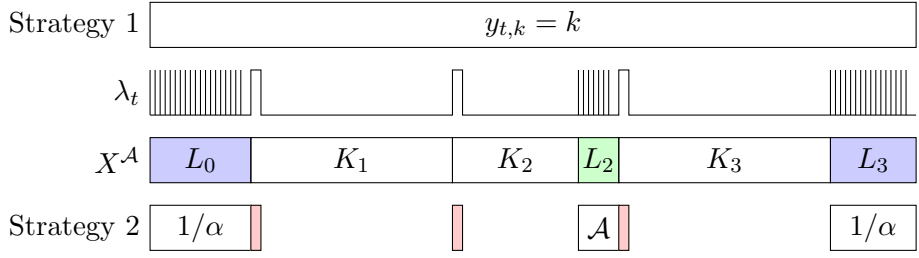
		\fi
	
	To estimate the cost of an optimal schedule we consider two strategies (see Figure~\ref{fig:online:lower:induction}): In the first strategy, a server of type $k$ is powered up at the first time slot and runs for the whole time except for phases $K_i$ with $\delta_i > 1$, then powering down and powering up are cheaper than keeping the server in the active state ($\beta_k$ vs. $\delta_i \beta_k$). 
	The operating cost for the phases $K_i$ is $\delta^\ast_i \beta_k$ with $\delta^\ast_i \coloneqq \min \{1, \delta_i\}$ and the operating cost for the phases $L_i$ is at most $\frac{1}{N^2} p_i \beta_k$, because algorithm~$\mathcal{A}$ uses servers whose types are lower than $k$ and therefore the operating cost of $\mathcal{A}$ is at least $N^2$ times larger. Thus, the total cost of this strategy is upper bounded by
	\begin{equation}\label{eqn:online:lower:induction:opt:strategy:a}
	C(X^\ast) \leq \beta_k \left(1 + \sum_{i=1}^n \delta^\ast_i  + \sum_{i \in \mathcal{L} \cup \mathcal{S}} \frac{1}{N^2} p_i \right).
	\end{equation} 
	
	In the second strategy, for the long phases $L$ we use the strategy given by our induction hypothesis, while for the short phases $S$ we behave like algorithm~$\mathcal{A}$ and in the phases $K_i$ we run the server type 1 for exactly one time slot (note that in $K_i$ we only have $\lambda_t = 1$ in the first time slot of the phase). Therefore the total cost is at most 
	\begin{equation}\label{eqn:online:lower:induction:opt:strategy:b}
	C(X^\ast) \leq \beta_k \left( \sum_{i \in \mathcal{L}} \frac{1}{\alpha} p_i + \sum_{i \in \mathcal{S}} p_i + 2n \beta_1 / \beta_k \right)
	\end{equation}
	with $\alpha \coloneqq 2k - 2 - \epsilon_{k-1}$.
	
	The total cost of $\mathcal{A}$ is equal to $\beta_k \left( \sum_{i=1}^n (1 + \delta_i) + \sum_{i \in \mathcal{L} \cup \mathcal{S}} p_i \right)$, so the competitive ratio is given by 
	\begin{align*}
	\frac{C(X^\mathcal{A})}{C(X^\ast)} &\geq \frac{\sum_{i=1}^n (1 + \delta_i) + \sum_{i \in \mathcal{L} \cup \mathcal{S}} p_i}{C(X^\ast) / \beta_k} \\
	&= \frac{1 + \sum_{i=1}^n \delta_i + \sum_{i \in \mathcal{L} \cup \mathcal{S}} \frac{1}{N^2} p_i}{C(X^\ast) / \beta_k} \\
	&\qquad + \frac{n-1 + \sum_{i \in \mathcal{L} \cup \mathcal{S}} p_i \left(1 - \frac{1}{N^2}\right)}{C(X^\ast) / \beta_k} \\
	&\geq 1 + \frac{n-1 + \sum_{i \in \mathcal{L} \cup \mathcal{S}} p_i \left(1 - \frac{1}{N^2}\right)}{C(X^\ast) / \beta_k} .
	\end{align*}
	In the first step, the numerator is separated into two parts. Then $C(X^\ast)$ in the first fraction is estimated by equation~\eqref{eqn:online:lower:induction:opt:strategy:a} (first strategy). In the next step, we transform the second fraction. 
	\begin{align*}
	%\allowdisplaybreaks[2]
	\frac{C(X^\mathcal{A})}{C(X^\ast)}  \geq 1 &+ \frac{\left(\sum_{i\in \mathcal{L}} p_i + \sum_{i \in \mathcal{S}} \alpha p_i + 2n \alpha \frac{\beta_1}{\beta_k}\right)\left(1 - \frac{1}{N^2}\right)}{C(X^\ast) / \beta_k} \\ 
	&- \frac{\sum_{i \in \mathcal{L} \cup \mathcal{S}} \frac{1}{N^2}p_i}{C(X^\ast) / \beta_k} 
	+\frac{\sum_{i \in \mathcal{L} \cup \mathcal{S}} \frac{1}{N^2}p_i}{C(X^\ast) / \beta_k}\\
	&+ \frac{n-1 - \left(\sum_{i \in \mathcal{S}} (\alpha - 1) p_i + 2n \alpha \frac{\beta_1}{\beta_k}\right)\left(1 - \frac{1}{N^2}\right)}{C(X^\ast) / \beta_k} \\ 
	%\end{align*}\begin{align*} %manual page break
	\phantom{\frac{C(X^\mathcal{A})}{C(X^\ast)}} 
	\geq 1 &+ \alpha \left(1 - \frac{1}{N^2}\right) - \frac{\alpha}{N^2} 
	+ \frac{\sum_{i \in \mathcal{L} \cup \mathcal{S}} \frac{1}{N^2}p_i}{C(X^\ast) / \beta_k} \\ 
	&+ \frac{n-1 - \left(\sum_{i \in \mathcal{S}} (\alpha - 1) p_i + 2n \alpha \frac{\beta_1}{\beta_k}\right)\left(1 - \frac{1}{N^2}\right)}{C(X^\ast) / \beta_k} .\numberthis\label{eqn:online:lower:induction:calc:a}
	\end{align*}
	The last inequality uses equation~\eqref{eqn:online:lower:induction:opt:strategy:b} (second strategy) to estimate $C(X^\ast)$. In particular, we have 
	\begin{equation*}
	- \frac{\sum_{i \in \mathcal{L} \cup \mathcal{S}} \frac{1}{N^2}p_i}{\sum_{i \in \mathcal{L}} \frac{1}{\alpha} p_i + \sum_{i \in \mathcal{S}} p_i + 2n \beta_1 / \beta_k}  
	\geq -\frac{\sum_{i \in \mathcal{L} \cup \mathcal{S}} \frac{1}{N^2}p_i}{\sum_{i \in \mathcal{L} \cup \mathcal{S}} \frac{1}{\alpha}p_i} 
	= -\frac{\alpha}{N^2} .
	\end{equation*} 
	The fraction $\frac{\sum_{i \in \mathcal{L} \cup \mathcal{S}} \frac{1}{N^2}p_i}{C(X^\ast) / \beta_k}$ of~\eqref{eqn:online:lower:induction:calc:a} is transformed as follows
	\begin{align*}
	\frac{C(X^\mathcal{A})}{C(X^\ast)} \geq 1 &+ \alpha \left(1 - \frac{2}{N^2}\right) \\
	&+ \frac{\left(n + 1 + \sum_{i \in \mathcal{L} \cup \mathcal{S}} \frac{1}{N^2} p_i\right)\left(1 - \xi\right)}{C(X^\ast) / \beta_k} \\
	&+ \frac{\left(n + 1 + \sum_{i \in \mathcal{L} \cup \mathcal{S}} \frac{1}{N^2} p_i\right)\xi}{C(X^\ast) / \beta_k} \\
	&- \frac{2 + \left(\sum_{i \in \mathcal{S}} (\alpha - 1) p_i  + 2n\alpha \frac{\beta_1}{\beta_k}\right) \left(1 - \frac{1}{N^2}\right)}{C(X^\ast) / \beta_k}
	\end{align*}
	with $0 < \xi < 1$. By using equation~\eqref{eqn:online:lower:induction:opt:strategy:a} and $\delta^\ast_i \leq 1$ for all $i \in [n]$, we get
	\begin{align*}
	\frac{C(X^\mathcal{A})}{C(X^\ast)} &\geq 2 - \xi + \alpha \left(1 - \frac{2}{N^2}\right) %\\
	%&\phantom{{}\geq 2 - \xi}
	+ \frac{\left(n + 1 + \sum_{i \in \mathcal{L} \cup \mathcal{S}} \frac{1}{N^2} p_i\right)\xi}{C(X^\ast) / \beta_k} \\
	&\phantom{{}\geq 2 - \xi}- \frac{2 + \left(\sum_{i \in \mathcal{S}} (\alpha - 1) p_i  + 2n\alpha \frac{\beta_1}{\beta_k}\right) \left(1 - \frac{1}{N^2}\right) }{C(X^\ast) / \beta_k} \\
	&\geq 2 - \xi + \alpha \left(1 - \frac{2}{N^2}\right) %\\
	%&\phantom{{}\geq 2 - \xi}
	- \frac{-n\xi + 2 + \frac{2k(n+1)}{{N}} + \frac{4kn}{N^2}}{C(X^\ast) / \beta_k} .
	\end{align*}
	For the last estimation we used the following inequalities:
	\begin{align*}
	n + 1 + \sum_{i \in \mathcal{L} \cup \mathcal{S}} \frac{1}{N^2} p_i &\geq n, \\
	\alpha - 1 &\leq 2k, \\
	\sum_{i \in \mathcal{S}} p_i &\leq \frac{n+1}{N}  && \parbox{5.5cm}{(by $|\mathcal{S}| \leq n+1$ \\ \phantom{(}and $p_i \leq 1/N$ for $i \in \mathcal{S}$),} \\
	1 - \frac{1}{N^2} &\leq 1 \\
	\text{and} \qquad 2n\alpha \frac{\beta_1}{\beta_k} &\leq 4kn/N^2 && \textrm{(by $\alpha \leq 2k$, $k \geq 2$ and $\beta_k = N^{2k}$)}. \\
	\end{align*}
	With $N^2 \geq {N}$, $\xi \coloneqq 6k / {N}$, $\alpha \leq 2k$, the definition of $\alpha = 2k - 2 - \epsilon_{k-1}$ and $N \geq 2k$, we get
	\begin{align*}
	\frac{C(X^\mathcal{A})}{C(X^\ast)} &\geq 2k - \epsilon_{k-1} - \frac{10k}{{N}} - \frac{3}{C(X^\ast) / \beta_k} . \numberthis\label{eqn:online:lower:induction:calc:b}
	\end{align*}
	
	If $C(X^\ast) < \frac{{N}}{2k} \beta_k$ holds, then $C(X^\mathcal{A}) \geq {N} \beta_k$ (a precondition of Lemma~\ref{lemma:online:lower:induction}) implies $2k \cdot C(X^\ast) < {N} \beta_k \leq C(X^\mathcal{A})$, so equation~\eqref{eqn:online:lower:induction} is fulfilled and Lemma~\ref{lemma:online:lower:induction} holds. If $C(X^\ast) \geq \frac{{N}}{2k} \beta_k$, then $\frac{3}{C(X^\ast) / \beta_k} \leq \frac{6k}{{N}}$ and inequality~\eqref{eqn:online:lower:induction:calc:b} gives
	\begin{align*}
	\frac{C(X^\mathcal{A})}{C(X^\ast)} &\geq 2k - \epsilon_{k-1} - \frac{16k}{{N}} \\
	&\geq 2k - \frac{9(k-1)^2}{{N}} - \frac{16k}{{N}} \\
	&\geq 2k - \frac{9k^2}{{N}} \\
	&\geq 2k - \epsilon_k . \qedhere
	\end{align*}
\end{proof}

\begin{theorem} \label{theo:online:lower}
 	There is no deterministic online algorithm for the data-center right-sizing problem with heterogeneous servers and time- and load-independent operating costs whose competitive ratio is smaller than $2d$.
\end{theorem}

\begin{proof}%[Proof of Theorem~\ref{theo:online:lower}]
	Assume that there is an $(2d-\epsilon)$-competitive deterministic online algorithm $\mathcal{A}$. Let $N \coloneqq \max \{6d , \lceil 9k^2 / \epsilon + 1\rceil \}$. We construct a workload as described at the beginning of Section~\ref{sec:online:lower} until the cost of $\mathcal{A}$ is greater than $N \beta_d$ (note that $l_j > 0$ for all $j \in [d]$,  so the cost of $\mathcal{A}$ can be arbitrarily large). By using Lemma~\ref{lemma:online:lower:induction} with $k = d$, we get
	\begin{align*}
	C(X^{\mathcal{A}}) &\geq (2d - \epsilon_d) \cdot C(X^\ast) \\
	&\geq \left(2d - \frac{9k^2}{\lceil 9k^2 / \epsilon + 1\rceil}\right) \cdot C(X^\ast) \\
	&> (2d - \epsilon) \cdot C(X^\ast) ,
	\end{align*}
	which is a contradiction to our assumption that algorithm $\mathcal{A}$ is $(2d-\epsilon)$\hyp{}competitive. Therefore, there is no deterministic online algorithm whose competitive ratio is smaller than $2d$. \qedllncs
\end{proof}

%\begin{proof}[Proof idea]
%	Assume that there is an $(2d-\epsilon)$-competitive deterministic online algorithm $\mathcal{A}$.  We construct a workload as described at the beginning of this section until the cost of $\mathcal{A}$ is greater than $N \beta_d$ (note that $l_j > 0$ for all $j \in [d]$,  so the cost of $\mathcal{A}$ can be arbitrarily large).
%	By using Lemma~\ref{lemma:online:lower:induction} with $k = d$ and $N \coloneqq \max \{6d , \lceil 9k^2 / \epsilon + 1\rceil \}$, we get $C(X^\mathcal{A}) > (2d - \epsilon) \cdot C(X^\ast)$ which is a contradiction to our assumption. See  the full paper for more details. \qedllncs
%\end{proof}

The schedule constructed for the lower bound only uses at most one job in each time slot, so there is no reason for an online algorithm to utilize more than one server of a specific type. Thus, for a data center with $m$ unique servers (i.e. $m_j = 1$ for all $j \in [d]$), the best achievable competitive ratio is $2d = 2m$. 

\begin{corollary}
	There is no deterministic online algorithm for the data-center right-sizing problem with $m$ unique servers and time- and load-independent operating costs whose competitive ratio is smaller than $2m$.
\end{corollary}

\section{Summary}
In this paper, we have settled the competitive ratio of online algorithms for right-sizing heterogeneous data centers with $d$ different server types. We investigated a basic setting where each server type has a constant operating cost per time unit. In contrast to related publications like~\cite{LinWierman2013} or~\cite{Sellke2020}, we studied the discrete setting where the number of active servers must be an integral number. Thereby we gain truly feasible solutions. 
We developed a $2d$-competitive deterministic online algorithm and showed that $2d$ is a lower bound for deterministic algorithms. Hence our algorithm is optimal. Furthermore, we presented a randomized version that achieves a competitive ratio of $\frac{e}{e-1}d \approx 1.582d$ against an oblivious adversary.

\appendix
\section{Variables}% and notation}
\label{sec:appendix:variables}

%Let $[k] \coloneqq \{1, 2, \dots k\}$, $[k]_0 \coloneqq \{0, 1, \dots k\}$ and $[k:l] \coloneqq \{k, k+1, \dots, l\}$ where $k,l \in \mathbb{N}_0$. 
The following table gives an overview of the variables defined in this paper.

\begingroup
	\setlength{\tabcolsep}{4pt}
	\renewcommand{\arraystretch}{1.3}
	\scriptsize
	
	\centering
	\begin{longtable}{|l|p{13cm}|}
		\hline
		\textbf{Variable} & \textbf{Description} \\
		\hline
		\endhead
		\hline
		\endfoot
		$A_{t,k}$ & Block at time slot $t$ in lane $k$ of the schedule $X^\mathcal{A}$ \\
		$\mathcal{A}$ & Our deterministic online algorithm (Section~\ref{sec:online:const}) or any online algorithm (Section~\ref{sec:online:lower}) \\
		$\beta_j$ & Switching cost of server type $j$ \\
		$C(X)$ & Total cost of the schedule $X$ (see equation~\eqref{eqn:online:cost}) \\
		$C_{t,k}(X)$ & Switching and operating cost of the schedule $X$ at time $t$ in lane $k$ (see equation~\eqref{eqn:online:const:ctkdef}) \\
		$C(A_{t,k})$ & Cost of block $A_{t,k}$ (see Lemma~\ref{lemma:online:const:alg:costai}) \\
		$d$ & Number of server types \\
		$d_{t,k}$ & Duration of block $A_{t,k}$. Formally, $d_{t,k} \coloneqq e_{t,k} - s_{t,k}$ \\
		$e_k$ & Variable in algorithm~$\mathcal{A}$ that stores the time slot when the server in lane $k$ will be powered down \\
		$e_{t,k}$ & Last time slot (exclusive) of block $A_{t,k}$, i.e. $e_{t,k}$ is the first time slot after $A_{t,k}$ \\
		$\mathcal{I}$ & Problem instance. Formally,  $\mathcal{I} \coloneqq (T, d, \vec{m}, \vec{\beta}, \vec{l}, \Lambda)$ \\
		$\mathcal{I}^t$ & Problem instance that ends at time slot $t$. Formally, $\mathcal{I}^t \coloneqq (t, d, \vec{m}, \vec{\beta}, \vec{l}, \Lambda^t)$ \\
		$j_{t,k}$ & Server type used in block $A_{t,k}$ \\
		$m$ & Total number of servers, $m \coloneqq \sum_{j=1}^{d} m_j$\\
		$m_j$ & Number of servers of type $j$ \\
		$l_j$ & Operating cost of server type $j$ \\
		$\lambda_t$ & Job volume that arrives at time slot $t$ \\
		$s_{t,k}$ & First time slot of block $A_{t,k}$ \\
		$\bar{t}_j$ & Number of time slots that a server of type $j$ stays active in algorithm~$\mathcal{A}$; $\bar{t}_j \coloneqq \left\lfloor \beta_j / l_j \right\rfloor$ \\
		$T$ & Total number of time slots \\
		$X$ & An arbitrary schedule. Formally, $X = (\vec{x}_1, \dots, \vec{x}_T)$ and $\vec{x}_t = (x_{t,1}, \dots, x_{t,d})$ \\
		$X^\ast$ & An optimal schedule \\
		$X^\mathcal{A}$ & The schedule calculated by our deterministic online algorithm (in Section~\ref{sec:online:const}) or by any online algorithm (in Section~\ref{sec:online:lower}) \\
		$X^\mathcal{B}$ & The schedule calculated by our randomized online algorithm \\
		$\hat{X}^t$ & Optimal schedule for the problem instance $\mathcal{I}^t$ that ends at time $t$ \\
		$x_{t,j}$ & Number of active servers of type $j$ at time $t$ in the schedule $X$ \\
		$x^\mathcal{A}_{t,j}$ & Number of active servers of type $j$ at time $t$ in the schedule $X^\mathcal{A}$ \\
		$\hat{x}^u_{t,j}$ & Number of active servers of type $j$ at time $t$ in the schedule $\hat{X}^u$ \\
		$y_{t,k}$ & Server type used in the $k$-th lane at time $t$ in the schedule $X$ (see equation~\eqref{eqn:online:const:ytkdef}) \\
		$y^\mathcal{A}_{t,k}$ & Server type used in the $k$-th lane at time $t$ in the schedule $X^\mathcal{A}$ (see equation~\eqref{eqn:online:const:ytkdef}) \\
		$\hat{y}^u_{t,k}$ & Server type used in the $k$-th lane at time $t$ in the schedule $\hat{X}^u$ (see equation~\eqref{eqn:online:const:ytkdef}) \\
		$\tilde{y}^u_{t,k}$ & Largest server type used in lane $k$ by the schedule $\hat{X}^{t'}$ at time slot $t'$ for $t' \in [t:u]$. Formally, $\tilde{y}^u_{t,k} \coloneqq \max_{t' \in [t:u]} \hat{y}^{t'}_{t',k}$ \\
	\end{longtable}
\endgroup

\bibliographystyle{plainurl}
%\bibliography{../meta/literature}
\bibliography{literature}

\end{document}